\documentclass[12pt,a4paper]{article}
\usepackage[inner=1in,outer=1in,top=1in]{geometry}	%左边距inner，右边距outer，上边距top
\usepackage{amsmath}	
\usepackage[T1]{fontenc}
\usepackage[utf8]{inputenc}
\usepackage{authblk}
\allowdisplaybreaks[3] 				%导入数学公式包的前提下，允许跨页，1，2，3，4表示允许程度，但equation环境不可换页
\usepackage{amsfonts}
\usepackage{algorithmic}
\usepackage{algorithm}
\usepackage{array}
\usepackage[caption=false,font=normalsize,labelfont=sf,textfont=sf]{subfig}
\usepackage{textcomp}
\usepackage{stfloats}
\usepackage{url}
\usepackage{verbatim}
\usepackage{graphicx}
\usepackage{cite}
\usepackage{bbm} 					%粗体、双线体
\usepackage{multirow} 				%纵向合并单元格
\usepackage{booktabs} 				%三线表
\usepackage{color}
\usepackage{amssymb}				%数学定理和证明过程包

\usepackage{enumitem}				%调整垂直/水平间距/标签样式
\usepackage{titlesec}
\usepackage{amsthm} 				%定理环境
\usepackage{hyperref}
\usepackage{float}
\usepackage{makecell}
\usepackage{tablefootnote}
\usepackage{xcolor} % 用于定义颜色
\hypersetup{
	colorlinks=true,    % 设置链接颜色
	linkcolor=blue,     % 内部链接颜色
	citecolor=blue,     % 文献引用颜色
	urlcolor=blue       % URL 链接颜色
}
\newtheorem{Thm}{Theorem}[section] 		% 定理
\newtheorem{Lem}[Thm]{Lemma}	%引理
\newtheorem{Rmk}[Thm]{Remark}		%评论
\newtheorem{Exm}[Thm]{Example}
\newtheorem{problem}[Thm]{Open problem}

\newtheorem{Mainthm}[Thm]{Main theorem}
\newtheorem{Con}{Construction}
\newtheorem{Cond}{Condition}
\newtheorem{Fra}{Framework}
\newcommand{\NN}{\mathbb{N}}%Natural numbers
\newcommand{\F}{\mathbb{F}}%
\begin{document}
\title{Algebraic Geometry Codes for Distributed Matrix Multiplication Using Local Expansions}

%\author{Jiang Li, Songsong Li, Chaoping Xing}

\author[*]{Jiang Li}
\author[*]{Songsong Li}
\author[*]{Chaoping Xing}
\affil[*]{School of Electronic Information and Electrical Engineering, Shanghai Jiao Tong University, Shanghai,  China,\ \authorcr{lijiang22222@sjtu.edu.cn, songsli@sjtu.edu.cn, xingcp@sjtu.edu.cn }}

\iffalse
\address{School of Electronic Information and Electrical Engineering, Shanghai Jiao Tong University, Shanghai,  China}
\email{lijiang22222@sjtu.edu.cn}

\author{Songsong Li}
\address{School of Electronic Information and Electrical Engineering, Shanghai Jiao Tong University, Shanghai,  China}
\email{songsli@sjtu.edu.cn}

\author{Chaoping Xing}
\address{School of Electronic Information and Electrical Engineering, Shanghai Jiao Tong University, Shanghai,  China}
\email{xingcp@sjtu.edu.cn}
\fi

\date{}

\maketitle

\begin{abstract}
Code-based Distributed Matrix Multiplication (DMM) has been extensively studied in distributed computing for efficiently performing large-scale matrix multiplication using coding theoretic techniques. The communication cost and recovery threshold (i.e., the least number of successful worker nodes required to recover the product of two matrices) are two major challenges in coded DMM research. Several constructions based on Reed-Solomon (RS) codes are known, including Polynomial codes, MatDot codes, and PolyDot codes. The PolyDot code is a unified framework of Polynomial and MatDot codes, offering a trade-off between the recovery threshold and communication costs. However, these RS-based schemes are not efficient for small finite fields because the distributed order (i.e., the total number of worker nodes) is limited by the size of the underlying finite field. Algebraic geometry (AG) codes, as a generalization of RS codes, can have a code length exceeding the size of the finite field, which helps solve this problem. Some work has been done to generalize Polynomial and MatDot codes to algebraic geometry codes, but the generalization of PolyDot codes to algebraic geometry codes still remains an open problem as far as we know. This is because functions of an algebraic curve do not behave as nicely as polynomials.

In this work, by using local expansions of functions, we are able to generalize the  three DMM schemes based on Reed-Solomon codes to algebraic geometry codes. Specifically, we provide a construction of AG-based PolyDot codes for the first time. In addition, our AG-based Polynomial and MatDot codes achieve better recovery thresholds compared to previous AG-based DMM schemes while maintaining similar communication costs. Our constructions are based on a novel basis of the Riemann-Roch space using local expansions, which naturally generalizes the standard monomial basis of the univariate polynomial space in RS codes. In contrast, previous work used the non-gap numbers to construct a basis of the Riemann-Roch space, which can cause cancellation problems that prevent the conditions of PolyDot codes from being satisfied.

%The main novelty of our constructions is the use of some special divisors to enlarge one-point Riemann-Roch space and the construction of a novel basis of the enlarging Riemann-Roch space from local expansions of functions at any rational place.

 \end{abstract}

\section{Introduction}

Distributed computing is a method that utilizes multiple worker nodes working collaboratively to solve large-scale tasks. It offers several advantages over centralized computing, including high reliability, scalability, computational speed, and cost-effectiveness. This work focuses primarily on distributed matrix multiplication (DMM), a specialized form of distributed computing aimed at computing the product of two large-scale matrices. Matrix multiplication is a fundamental operation in various fields such as machine learning, scientific computing, and graph processing.

During DMM, worker nodes need to exchange data with other nodes, which introduces two significant challenges: communication overhead during interactions and the impact of straggler nodes that run slower or are prone to delays. Since distributed computing requires waiting for all worker nodes to complete their tasks, these straggler nodes can significantly affect the overall runtime of the computation. Recent works \cite{dutta2017coded,lee2017speeding,lee2017high,yu2017polynomial,dutta2019optimal,yu2020straggler,dutta2018unified,fidalgo2024distributed} have made efforts to reduce the communication load and alleviate the effects of straggler nodes in distributed matrix multiplication. Other studies \cite{CT18, MLG22, MSJ20, DEK20, OC24} explore secure distributed matrix multiplication, considering data privacy protection during distributed computing.

Code-based distributed matrix multiplication has been proposed as a promising solution to mitigate straggler effects. Coding theory facilitates theoretical analyses of the fundamental limits of these constructions and allows for explicit comparisons with existing strategies. In the framework of coded DMM, there are \(N\) worker nodes and a master node. The master node first encodes two multiplicand matrices, \(A\) and \(B\), into two codewords \(c_A\) and \(c_B\) of length \(N\), whose coordinates are small matrices. The \(i\)-th coordinates of \(c_A\) and \(c_B\) are then sent to the \(i\)-th worker nodes for distributed computing. Finally, the master node downloads the results from the worker nodes to recover the product \(AB\). In this context, the straggler nodes can be viewed as erasure errors, which can be corrected using erasure decoding techniques.

The best-known coded DMM schemes are constructed from Reed-Solomon codes. However, the number of worker nodes is constrained by the size of the underlying finite field \(\mathbb{F}_q\), i.e., \(N \leq q\). Consequently, the advantages of code-based DMM for large-scale matrices are limited, especially in small finite fields where RS-based schemes are impractical. Algebraic geometry (AG) codes, as a generalization of RS codes, can have a code length exceeding the size of the finite field, thus solving this problem. There has been some work on generalizing Polynomial codes and MatDot codes to algebraic geometry codes, but the generalization of PolyDot codes remains an open problem as far as we know.

%In this paper, we address this limitation by generalizing Reed-Solomon codes to Algebraic Geometry codes (AG codes), which can have a code length \(N\) larger than the field size.

\subsection{Related work}
Assume \(A \in \mathbb{F}_q^{t \times r}\) and \(B \in \mathbb{F}_q^{r \times s}\) are two large-scale matrices defined over a finite field \(\mathbb{F}_q\). In the series of works on code-based DMM, the matrices \(A\) and \(B\) are partitioned into \(mp\) and \(np\) submatrices as follows:
\begin{equation}\label{eq:1}
A =
\begin{pmatrix}
	A_{11} & \cdots & A_{1p} \\
	\vdots & \ddots & \vdots \\
	A_{m1} & \cdots & A_{mp}
\end{pmatrix}
, \quad
B =
\begin{pmatrix}
	B_{11} & \cdots & B_{1n} \\
	\vdots & \ddots & \vdots \\
	B_{p1} & \cdots & B_{pn}
\end{pmatrix}
,
\end{equation}
where \(m\), \(n\), and \(p\) are divisors of \(t\), \(s\), and \(r\), respectively. Then, the product \(AB\) is given by
\[
AB = \begin{pmatrix}
	C_{11} & \cdots & C_{1n} \\
	\vdots & \ddots & \vdots \\
	C_{m1} & \cdots & C_{mn}
\end{pmatrix},
\]
where \(C_{iw} = \sum_{j=1}^p A_{ij} B_{jw}\). Specifically, in Polynomial codes, \(p = 1\), while in MatDot codes, \(m = n = 1\). Thus, the product of \(A\) and \(B\) corresponds to an outer (resp. inner) product of matrix vectors in Polynomial (resp. MatDot) codes.

% During preprocessing, the master node first selects two polynomials \(f(x)\) and \(g(x)\) with \(\{A_{ij} \mid 1 \leq i \leq m, 1 \leq j \leq p\}\) and \(\{B_{kw} \mid 1 \leq k \leq p, 1 \leq w \leq n\}\) as their coefficients, respectively. It then encodes \(f(x)\) and \(g(x)\) into two codewords by evaluating them at \(N\) distinct points in \(\mathbb{F}_q\). The master node sends the \(i\)-th coordinates of the two codewords to the \(i\)-th worker nodes. Each worker node performs the small matrix multiplications locally and sends the results back to the master node. Finally, the master node decodes the product of \(f(x)\) and \(g(x)\) from the received information to recover \(AB\). Therefore, the recovery threshold is \(\deg(f(x)g(x)) + 1\). The main challenge in constructing code-based DMM is to choose suitable \(f(x)\) and \(g(x)\) so that \(\deg(f(x)g(x)) + 1\) is minimized.
In \cite{yu2017polynomial}, Yu et al. introduced a novel code-based construction for DMM using Polynomial codes. Subsequently, two more constructions, MatDot and PolyDot codes, were introduced in \cite{dutta2019optimal} to improve the recovery threshold, i.e., the minimum number of worker nodes required to recover \(AB\). MatDot codes achieve a better recovery threshold compared to Polynomial codes but incur higher communication overhead. PolyDot codes provide a unified construction that combines Polynomial and MatDot codes, offering a trade-off between recovery threshold and communication costs. Meanwhile, Entangled Polynomial codes \cite{yu2020straggler} and Generalized PolyDot codes \cite{dutta2018unified} have been introduced to achieve better recovery thresholds compared to PolyDot codes. All of these polynomial-based codes are subcodes of Reed-Solomon codes (RS codes). Therefore, we collectively refer to them as RS-based distributed matrix multiplication (RS-based DMM).

%We provide a general framework for RS-based DMM in Section \hyperref[sec:2.4]{2.4}. 
For comparison with our results, we summarize the corresponding recovery threshold, communication cost, and computation complexity of worker nodes in known RS-based DMM schemes in the following table.
The cost is measured in terms of the number of elements or operations in the underlying finite field. The communication cost from the master node to the \(N\) worker nodes is referred to as the upload cost. The communication cost from the \(R\) successful worker nodes back to the master node is termed the download cost. The complexity for each worker node to compute the matrix multiplications is referred to as the worker computation complexity.

 %The computational complexity of each worker and the decoding complexity by the master node in the final step are measured in terms of the number of operations in the underlying finite field.

\begin{table}[H]
	\centering
 \small
	\begin{tabular}{|c|c|c|c|c|}
		\hline
 \textbf{RS-based} & \textbf{Recovery} & \multicolumn{2}{|c|}{\textbf{Communication}} & \textbf{Worker} \\

	\textbf{DMM}	 & \textbf{Threshold \(R\)} & {Upload} & {Download}  & \textbf{Computation} \\
		\hline
		Polynomial \cite{yu2017polynomial} & {\(mn\)} & {\(O\left((\frac{t}{m} + \frac{s}{n})rN\right)\)} & {\(O(\frac{ts}{mn}R)\)} & {\(O(\frac{trs}{mn})\)} \\
  \hline
		Matdot\cite{dutta2019optimal} & \(2p - 1\) & \(O(\frac{(t + s)r}{p}N)\) & \(O(tsR)\) & \(O(\frac{trs}{p})\)  \\
		\hline
		Polydot\cite{dutta2019optimal} & \((2p - 1)mn\) & \(O\left((\frac{tr}{mp} + \frac{rs}{np})N\right)\) & \(O(\frac{ts}{mn}R)\) & \(O(\frac{trs}{mpn})\)  \\
		\hline
	\end{tabular}
	\caption{Performance metrics of Polynomial, MatDot, and PolyDot codes}

	\label{tab:1}
\end{table}

In \cite{machado2023hera}, the authors extended Polynomial codes from Reed-Solomon codes to Hermitian codes. Meanwhile, in \cite{makkonen2023algebraic}, the authors generalized the approach from \cite{DEK20} to Hyper-elliptic codes. However, these constructions focus on secure distributed matrix multiplication, which is outside the scope of this work. Recently, in \cite{fidalgo2024distributed}, the authors extended both Polynomial and MatDot codes from Reed-Solomon codes to one-point algebraic geometry codes, referred to as AG-based Polynomial codes and AG-based MatDot codes, respectively.
They also investigated the lower bounds of the optimal recovery thresholds. In their constructions, the AG-based Polynomial code achieves a recovery threshold close to the optimal bound \cite[Proposition 1]{fidalgo2024distributed}, namely \(mn + g\), where \(g\) denotes the genus of the underlying algebraic function field. The AG-based MatDot DMM is an optimal construction under specific conditions involving complex parameters related to the algebraic function field. For simplicity, we summarize their results for some special AG codes in the following table, where \(c(P):=\min\{s \in W(P) \mid [s,\infty) \subset W(P)\}\) and \(W(P)\) is the Weierstrass semi-group of the rational place \(P\). For any function field with genus \(g > 0\), it always holds that \(g < c(P) \leq 2g\).

\begin{table}[H]
	\centering
 \small
	\begin{tabular}{|c|c|c|c|c|}
		\hline
 \textbf{AG-based} & {\textbf{Recovery}} & \multicolumn{2}{|c|}{\textbf{Communication}} & \textbf{Worker} \\

	\textbf{DMM}& \textbf{Threshold R}  & {Upload} & {Download}  & \textbf{Computation} \\
		\hline
		Polynomial \cite{fidalgo2024distributed} & {\(\geq mn+c(P)\)}  & {\(O\left((\frac{t}{m} + \frac{s}{n})rN\right)\)} & {\(O(\frac{ts}{mn}R)\)} & {\(O(\frac{trs}{mn})\)} \\
  \hline
		Sparse MatDot\cite{fidalgo2024distributed}\tablefootnote{See \cite[Definition 13]{fidalgo2024distributed} for the definition of the sparse Weierstrass semi-group.
} & \(2p - 1+2c(P)\)   & \(O(\frac{(t + s)r}{p}N)\) & \(O(tsR)\) & \(O(\frac{trs}{p})\)  \\
  		\hline
 Elliptic MatDot\cite{fidalgo2024distributed} & $2p-1+2g+2$ & \(O(\frac{(t + s)r}{p}N)\) & \(O(tsR)\) & \(O(\frac{trs}{p})\)  \\ 	\hline
  Hermitian MatDot\cite{fidalgo2024distributed} & $2p-1+3g$ & \(O(\frac{(t + s)r}{p}N)\) & \(O(tsR)\) & \(O(\frac{trs}{p})\)  \\
		\hline
	\end{tabular}
	\caption{AG-based Polynomial and MatDot DMM in \cite{fidalgo2024distributed}}
	\label{tab:2}
\end{table}

Based on \cite{fidalgo2024distributed}, we propose the following two open problems:

\begin{problem}\label{prob:1}
    Can we construct AG-based Polynomial DMM schemes that achieve the optimal recovery threshold bound \(g + mn\) given in \cite[Proposition 1]{fidalgo2024distributed}?
\end{problem}

\begin{problem}\label{prob:2}
    Can we generalize the construction of PolyDot DMM from Reed-Solomon codes to algebraic geometry codes?
\end{problem}

	%The  The first Algebraic geometry code-based DMM was proposed in \cite{machado2023hera}. There are also many recent works using the idea of algebraic function fields to reduce the field size. For secure DMM, \cite{machado2023hera, makkonen2023algebraic} extend the approach initially introduced in \cite{d2020gasp}. Additionally, for DMM, \cite{fidalgo2024distributed} builds upon the foundations laid by \cite{yu2017polynomial} and \cite{dutta2019optimal}. 
\subsection{Our contributions}
\subsubsection{Main techniques}
For constructions of distributed matrix multiplication based on algebraic geometry codes, we first need a basis of the corresponding Riemann-Roch space for encoding and decoding. In particular, in the RS-based DMM, the Riemann-Roch space is simply the univariate polynomial space $\mathbb{F}_q[x]$, and the corresponding basis consists of the standard monomials $\{1, x, x^2, \dots\}$. In AG-based DMM, we first construct a novel basis of the Riemann-Roch space from local expansions, which naturally generalizes the standard monomial basis. As a result, we can generalize all RS-based DMM schemes to AG codes.

In previous work, the basis was constructed from the Weierstrass semigroup \(W(P)\) of some rational place \(P\) in the algebraic function field. This approach can cause cancellation problems, rendering the PolyDot codes ineffective. To avoid cancellation problems in our construction, we use some special divisors rather than one-point divisors to enlarge the Riemann-Roch space. This approach ensures that, on the one hand, the local expansions of the basis have sufficiently large gaps to prevent cancellations in the product of any two basis elements; on the other hand, these special divisors are well-suited for achieving better recovery thresholds in our AG-based DMM.

%In this work, we first construct a novel basis of the Riemann-Roch space from local expansions which is a nature generalization of the standard monomial basis in RS codes. In previous work, the basis is constructed from the non-gap numbers in $W(P)$ for some rational place $P$ of $F$, which will cause cancellation problems such that the conditions in the Polydot codes cannot be satisfied. To avoid the cancellation problem, we use  the one-point Riemann-Roch space he use of some  and the construction of a novel basis of the enlarging Riemann-Roch space from local expansions of functions at any rational

\subsubsection{Main results}
We generalize Polynomial codes \cite{yu2017polynomial}, MatDot codes \cite{dutta2019optimal}, and PolyDot codes \cite{dutta2019optimal} (as well as Entangled Polynomial codes \cite{yu2020straggler} and Generalized PolyDot codes \cite{dutta2018unified}) from Reed-Solomon codes to algebraic geometry codes based on general algebraic function fields using local expansions. Our constructions provide affirmative answers to Open Problems \ref{prob:1} and \ref{prob:2}. Given two matrices \(A \in \mathbb{F}_q^{t \times r}\) and \(B \in \mathbb{F}_q^{r \times s}\), which are partitioned into \(mp\) and \(np\) submatrices as shown in Equation~\eqref{eq:1}, our main results can be summarized as follows:

\begin{Mainthm}\label{thm:1.1}(AG-based Polynomial DMM)
	Assume \(F/\mathbb{F}_q\) is an arbitrary algebraic function field with genus \(g\). Let \(P\) be a rational place in \(F\) and \(W(P)\) the Weierstrass semigroup of \(P\). Then, Construction \hyperref[con:1]{1} in Section \hyperref[sec:3]{3} provides an AG-based Polynomial DMM scheme over \(\mathbb{F}_q\) with a recovery threshold of \(R = 2g + mn\). Moreover, if \(m \in W(P)\) or \(n \in W(P)\), then Construction \hyperref[con:2]{2} provides an AG-based Polynomial DMM scheme over \(\mathbb{F}_q\) with a recovery threshold of \(R = g + mn\).

\end{Mainthm}

As shown in Table \hyperref[tab:2]{2}, the recovery threshold \(R\) of AG-based Polynomial DMM in \cite{fidalgo2024distributed} satisfies \(R \geq mn + c(P)\), with equality holding when \(m \in W(P)\). Since \(g < c(P) \leq 2g\) for \(g > 0\), our constructions for AG-based Polynomial DMM achieve better recovery thresholds than those in \cite{fidalgo2024distributed}. Specifically, if the partition parameters \(m \in W(P)\) or \(n \in W(P)\), our construction attains the optimal bound given in \cite{fidalgo2024distributed}.

\begin{Mainthm}\label{thm:1.2}(AG-based MatDot DMM)
	Assume \(F/\mathbb{F}_q\) is an arbitrary algebraic function field with genus \(g\). Then, Construction \hyperref[con:3]{3} in Section \hyperref[sec:4]{4} is an AG-based MatDot DMM scheme defined over \(\mathbb{F}_q\) with a recovery threshold of \(R = 2g + 2p - 1\).

\end{Mainthm}

The optimal recovery threshold for AG-based MatDot codes in \cite{fidalgo2024distributed} is hard to determine as it depends on the structure of \(W(P)\). For some special AG codes presented in Table \hyperref[tab:2]{2}, our AG-based MatDot DMM achieves better thresholds than those in \cite{fidalgo2024distributed}. 
%It is worth mentioning that the optimal recovery threshold of AG-based MatDot codes given in \cite[Theorem 2]{fidalgo2024distributed} is defined over a one-point Riemann-Roch space, which is not the case in our constructions.

We also provide a construction for AG-based PolyDot DMM:

\begin{Mainthm}\label{thm:1.3}(AG-based PolyDot DMM)
	Assume \(F/\mathbb{F}_q\) is an arbitrary algebraic function field with genus \(g\). Then, Construction \hyperref[con:4]{4} in Section \hyperref[sec:5]{5} provides an AG-based PolyDot DMM scheme over \(\mathbb{F}_q\) with a recovery threshold of
\[
R = \left\{
\begin{array}{ll}
    4g + (2p-1)mn + 2mn - 2m & \text{if } m = 1 \text{ or } m \ge n \ge 2, \\
    4g + (2p-1)mn + 2mn - 2n & \text{if } n = 1 \text{ or } n > m \ge 2.
\end{array}
\right.
\]

\end{Mainthm}

The recovery threshold reflects the number of successful worker nodes during distributed computing. We also analyze the communication cost and worker computation complexity for the above constructions. In the following table, let \(R\) be the corresponding recovery threshold in Main Theorems~\ref{thm:1.1}, \ref{thm:1.2}, and \ref{thm:1.3}, respectively.

\begin{table}[H]
	\centering
	\begin{tabular}{|c|c|c|c|}
		\hline
 \textbf{AG-based} &  \multicolumn{2}{|c|}{\textbf{Communication}} &\textbf{Worker} \\

	\textbf{DMM}	 & {Upload} & {Download} & \textbf{Computation} \\
		\hline
		Polynomial[Thm.\ref{thm:1.1}]  & \(O\left((\frac{t}{m} + \frac{s}{n})rN\right)\) & \(O(\frac{ts}{mn}R)\) & \(O(\frac{trs}{mn})\)  \\
		\hline
		MatDot[Thm.\ref{thm:1.2}] & \(O(\frac{(t + s)r}{p}N)\) & \(O(tsR)\) & \(O(\frac{trs}{p})\)   \\
		\hline
		PolyDot[Thm.\ref{thm:1.3}] & \(O\left((\frac{tr}{mp} + \frac{rs}{np})N\right)\)  & \(O(\frac{ts}{mn}R)\) & \(O(\frac{trs}{mpn})\)  \\
		\hline
	\end{tabular}
	\caption{Complexities of AG-based Polynomial, Matdot, and PolyDot Codes}
	\label{tab:3}
\end{table}

Note that in RS-based DMM, the number of worker nodes \(N\) is limited by the size of the underlying finite field. For DMM involving large-scale matrices over \(\mathbb{F}_q\) with a distributed order \(N > q\), RS-based schemes require implementation in an extended finite field of size at least \(N\). In contrast, AG-based DMM allows us to choose an algebraic function field \(F/\mathbb{F}_q\) with enough rational points to accommodate \(N\). Thus, all AG-based DMM constructions can be implemented directly in \(\mathbb{F}_q\). Under the same distributed order \(N\) and \(N > q\), the bit complexities of upload and computation for each worker in our AG-based constructions are less than those in RS-based constructions by a factor of \(\frac{\log q}{\log N}\).
Moreover, if the genus \(g\) of an algebraic function field \(F\) satisfies the following conditions:
\begin{itemize}
    \item \(g < mn \frac{\log N / q}{\log q}\) in Polynomial DMM;
    \item \(g < (2p-1) \frac{\log N / q}{2 \log q}\) in Matdot DMM;
    \item \(g < (2p-1)mn\left(\frac{\log N / q}{4 \log q} - \frac{1}{4p-2}\right)\) in PolyDot DMM,
\end{itemize}
then the bit complexities of download in our AG-based DMMs are also better than those in RS-based DMMs.

\subsection{Organization}
This paper is organized as follows. In Section 2, we present some preliminaries on algebraic function fields over finite fields, algebraic geometry codes, local expansions, and the general framework of code-based DMM. In Sections 3, 4, and 5, we first review the classical Polynomial, MatDot, and PolyDot codes, and then introduce our constructions based on AG codes. In Section 6, we explain the decoding procedures for these AG-based DMM schemes. In Section 7, we compare the recovery thresholds of our AG-based Polynomial and MatDot codes with those in \cite{fidalgo2024distributed}.

\section{Preliminaries}
In this section, we present some preliminaries on algebraic function fields over finite fields, algebraic geometry codes, local expansions and the the general framework of code-based DMM.

\subsection{Algebraic function fields over finite fields}  \label{sec:2.1}
Let us introduce some basic notations and facts about algebraic function fields. For more details, we refer the reader to \cite{stichtenoth2009algebraic}.

Let $q$ be a prime power and $\mathbb{F}_q$ be the finite field with $q$ elements. Let $F/\mathbb{F}_q$ be an algebraic function field with the full constant field $\mathbb{F}_q$. Let $\mathbb{P}_F$ denote the set of places of $F$ and let $g(F)$ denote the genus of $F$. The degree $\deg(P)$ of a place $P$ is the degree of the residue class field of $P$ over $\mathbb{F}_q$. Recall that a place $P$ of degree 1 is also called a rational place.

For a given place $P$, let $\nu_P$ denote the corresponding normalized discrete valuation. For a function $f \in F$ and a place $P \in \mathbb{P}_F$, we say $P$ is a zero (resp., pole) of $f$ if $\nu_P(f) > 0$ (resp., $\nu_P(f) < 0$). A divisor $G$ of $F$ is a formal sum $G = \sum_{P \in \mathbb{P}_F} n_P P$ with only finitely many nonzero coefficients $n_P \in \mathbb{Z}$. The support of $G$ is defined as $\text{supp}(G) = \{ P \in \mathbb{P}_F : n_P \neq 0 \}$ and the degree of $G$ is defined as $\deg(G) = \sum_{P \in \mathbb{P}_F} n_P \deg(P)$. We denote the divisor group of $F$ by $\text{Div}(F)$.
For any nonzero element $f \in F$, the zero divisor of $f$ is defined by $(f)_0 = \sum_{P \in \mathbb{P}_F, \nu_P(f) > 0} \nu_P(f) P$, and the pole divisor of $f$ is defined by $(f)_\infty = \sum_{P \in \mathbb{P}_F, \nu_P(f) < 0} -\nu_P(f) P$. The principal divisor of $f$ is then
\[
(f) := (f)_0 - (f)_\infty = \sum_{P \in \mathbb{P}_F} \nu_P(f) P.
\]

For any divisor $G \in \text{Div}(F)$, the Riemann-Roch space associated with $G$ is defined by
\[
\mathcal{L}(G) = \{f \in F \setminus \{0\} : (f) \geq -G\} \cup \{0\}.
\]
It is a finite-dimensional vector space over $\mathbb{F}_q$, whose dimension is denoted by $\ell(G)$. From Riemann's theorem \cite[Theorem 1.4.17]{stichtenoth2009algebraic}, $\ell(G) \geq \deg(G) - g + 1$. Recall that if $\deg(G) < 0$, then $\ell(G) = 0$, and if $\deg(G) \ge 2g - 1$, then $\ell(G) = \deg(G) - g + 1$.

For a given place $P$, denote $W(P)$ as the Weierstrass semigroup of $P$, which is a sub-semigroup of the additive semigroup $\mathbb{N}$:
\[
W(P) = \left\lbrace k \in \mathbb{N} : \mathcal{L}(kP) \not= \mathcal{L}((k-1)P) \right\rbrace.
\]
An integer $k \in \mathbb{N}$ is called a pole number of $P$ if $k \in W(P)$; otherwise, $k$ is called a gap number of $P$. By the Weierstrass gap theorem \cite[Theorem 1.6.8]{stichtenoth2009algebraic}, if $g > 0$, there are exactly $g$ gap numbers $i_1 < \ldots < i_g$ of $P$ with $i_1 = 1$ and $i_g \le 2g - 1$.

\subsection{Algebraic geometry codes}  \label{sec:2.2}
Let $F/\mathbb{F}_q$ be an algebraic function field with the full constant field $\mathbb{F}_q$. Let $\mathcal{P} = \{P_1, \ldots, P_n\}$ be a set of $n$ distinct rational places of $F$. For a divisor $G \in \text{Div}(F)$ with $\text{supp}(G) \cap \mathcal{P} = \emptyset$, the evaluation map is defined as
\[\text{ev}_{\mathcal{P}}: \mathcal{L}(G) \to \mathbb{F}_q^n,\ f \mapsto (f(P_1), \ldots, f(P_n)),\]
which is well-defined since the $P_i$'s are not in the support of $G$. Note that $\text{ev}_{\mathcal{P}}$ is an $\mathbb{F}_q$-linear map with kernel $\ker(\text{ev}_{\mathcal{P}}) = \mathcal{L}(G - \sum_{i=1}^n P_i)$. If $\deg(G) < n$, then $\text{ev}_{\mathcal{P}}$ is injective. Each algebraic function $f \in \mathcal{L}(G)$ is uniquely determined by its evaluations at any $\deg(G)+1$ rational places among $P_1, P_2, \ldots, P_n$.

The algebraic geometry code associated with $\mathcal{P}$ and $G$ is defined by
\[
C_\mathcal{L}(\mathcal{P}, G) := \text{ev}_{\mathcal{P}}(\mathcal{L}(G)) = \{(f(P_1), f(P_2), \ldots, f(P_n)) : f \in \mathcal{L}(G)\}.
\]

Suppose that $0 < \deg(G) < n$. Then $C_\mathcal{L}(\mathcal{P}, G)$ is an $[n, k, d]$-linear code with dimension $k = \ell(G)$ and minimum distance $d \ge n - \deg(G)$ by \cite[Theorem 2.2.2]{stichtenoth2009algebraic}.

\subsection{Local expansions}\label{sec:2.3}
In this subsection, we discuss local expansions and present the key lemma necessary for the constructions in Sections \hyperref[sec:3]{3} and \hyperref[sec:4]{4}.

Let $F/\mathbb{F}_q$ be an algebraic function field with the full constant field $\mathbb{F}_q$, and let $P$ be a rational place. An element $\tau$ of $F$ is called a local parameter at $P$ if $\nu_P(\tau) = 1$ (such a local parameter always exists). For a nonzero function $f \in F$ with $\nu_P(f) \ge v$, the local expansion of $f$ at $P$ is 
\[
f = \sum_{r=v}^\infty f_r \tau^r, \tag{2} \label{eq:2}
\]
where all $f_r\in\F_q$. The local expansion can be computed as the following procedure:
\begin{itemize}
    \item Let $f_v = \left(\frac{f}{\tau^v}\right)(P)$, i.e., $f_v$ is the evaluation of the function $f/\tau^v$ at $P$. Then $\nu_P \left(f - f_v \tau^v \right) \ge v + 1$.
    \item For $m\in[v,\infty)$, let $f_{m+1} = \left(\frac{f - \sum_{r=v}^m f_r \tau^r}{\tau^{m+1}}\right)(P)$. By induction, $\nu_P \left(f - \sum_{r=v}^{m+1} f_r \tau^r \right) \ge m + 2$. 
\end{itemize}
Thus we can obtain an infinite sequence $\{f_r\}_{r=v}^\infty$ of elements of $\mathbb{F}_q$ such that $$\nu_P \left(f - \sum_{r=v}^m f_r \tau^r \right) \ge m + 1.$$
\iffalse
we have $\nu_P \left(\frac{f}{t^v}\right) \ge 0$. 

 Then the function $f/t^v - f_v$ satisfies $\nu_P \left(\frac{f}{t^v} - f_v \right) \ge 1$, i.e., $\nu_P \left(f - f_v t^v\right) \ge v+1$. Define $f_{v+1} = \left(\frac{f - f_v t^v}{t^{v+1}}\right)(P)$. Then $\nu_P(f - f_v t^v - f_{v+1} t^{v+1}) \ge v + 2$.

Assume that we have obtained a sequence $\{f_r\}_{r=v}^m$ ($m > v$) of elements of $\mathbb{F}_q$ such that $\nu_P \left(f - \sum_{r=v}^k f_r t^r \right) \ge k + 1$ for all $v \le k \le m$. Define . Then $\nu_P \left(f - \sum_{r=v}^{m+1} f_r t^r \right) \ge m + 2$. In this way, by continuing our construction of $f_r$, we can obtain an infinite sequence $\{f_r\}_{r=v}^\infty$ of elements of $\mathbb{F}_q$ such that $\nu_P \left(f - \sum_{r=v}^m f_r t^r \right) \ge m + 1$ for all $m \ge v$. Summarizing the above construction, we get

which is called the local expansion of $f$ at $P$.
\fi
%It is clear that the local expansion of a function depends on the choice of the local parameter $t$. Note that if a power series $\sum_{i=v}^\infty f_i t^i$ satisfies $\nu_P \left(f - \sum_{i=v}^m f_i t^i \right) \ge m + 1$ for all $m \ge v$, then it is a local expansion of $f$. 
The above procedure shows that finding a local expansion at a rational place is very efficient as long as the computation of evaluations of functions at this place is easy.

In the following, let $D \in \text{Div}(F)$ be a positive non-special divisor with $\deg(D) = g$ and $\ell(D) = 1$. The existence of $D$ is guaranteed by \cite[Proposition 2]{niederreiter1999new}. Given a positive divisor $A \in \text{Div}(F)$, we consider the Riemann-Roch space $\mathcal{L}(D + A)$. The following lemma will be useful to determine its dimension.

\begin{Lem} \label{lem:2.1}
Let $D \in \text{Div}(F)$ be a positive non-special divisor with $\deg(D) = g$ and $\ell(D) = 1$. Then for any positive divisor $A \in \text{Div}(F)$, the Riemann-Roch space $\mathcal{L}(D + A)$ has dimension $\ell(D + A) = \deg(A) + 1$.

\end{Lem}

\begin{proof}
First, from Riemann's theorem \cite[Theorem 1.4.17]{stichtenoth2009algebraic}, we have $\ell(D + A) \ge \deg(D + A) + 1 - g = \deg(A) + 1$. On the other hand, we have $\ell(D + A) \le \ell(D) + \deg(D + A) - \deg(D) = \deg(A) + 1$ by \cite[Lemma 1.4.8]{stichtenoth2009algebraic}. Thus, $\ell(D + A) = \deg(A) + 1$.
\end{proof}

%Let $D \in \text{Div}(F)$ be a positive non-special divisor with $\deg(D) = g$ and $\ell(D) = 1$ such that $P \not\in \text{supp}(D)$. Then, for any integer $v \in \mathbb{N}$, we can construct a basis for the Riemann-Roch space $\mathcal{L}(D + vP)$ by using local expansions. In particular, we have the following lemma, which is crucial for the constructions in the subsequent sections:
Particularly, for $A=vP$ and $P$ is a rational place such that $P \not\in \text{supp}(D)$, we can construct a basis of $\mathcal{L}(D + vP)$ which will be used in our AG-based constructions for DMM.

\begin{Lem} \label{lem:2.2}
Let $D \in \text{Div}(F)$ be a positive non-special divisor with $\deg(D) = g$ and $\ell(D) = 1$. Assume $P \in \mathbb{P}_F$ is a rational place such that $P \not\in \text{supp}(D)$. Let $\tau \in F$ be a local parameter at $P$, i.e., $\nu_P(\tau) = 1$. Then, for any integer $v \in \mathbb{N}$, there exists a basis $\{\hat{f}_0, \hat{f}_1, \ldots, \hat{f}_v\}$ of the Riemann-Roch space $\mathcal{L}(D + vP)$ satisfying $\hat{f}_0 = 1 \in \mathbb{F}_q$ and
\[
\hat{f}_i = \tau^{-i} + \sum_{\ell=1}^{\infty} \lambda_\ell^{(i)} \tau^\ell \quad \text{for} \quad i = 1, \ldots, v.
\]
% for $i = 1, \ldots, v$. %Observe that $\hat{f}_i \in \mathcal{L}(D + iP) \backslash \mathcal{L}(D + (i-1)P)$ for $i = 1, \ldots, v$.
\end{Lem}

\begin{proof}
For any function $\hat{f} \in \mathcal{L}(D + vP)$, we have $\nu_P(f) \ge -v$. By Equation~\eqref{eq:2}, assume the local expansion of $\hat{f}$ at $P$ is
\[
\hat{f} = \sum_{\ell=-v}^{\infty} \lambda_\ell \tau^\ell, \quad \lambda_\ell \in \mathbb{F}_q.
\]
Consider the following mapping $\phi: \mathcal{L}(D + vP) \to \mathbb{F}_q^{v}$ given by 
\[
\hat{f} = \sum_{\ell=-v}^{\infty} \lambda_\ell \tau^\ell \longmapsto (\lambda_{-1}, \lambda_{-2}, \ldots, \lambda_{-v}).
\]
Observe that $\phi$ is an $\mathbb{F}_q$-linear map with kernel $\ker(\phi) = \mathcal{L}(D)$. Since $\ell(D + vP) - \ell(D) = v$ by Lemma \hyperref[lem:2.2]{2.2}, the mapping $\phi$ is surjective.
For each vector $e_i=(0, \ldots, 1, \ldots, 0) \in \mathbb{F}_q^{v}$, where only the $i$-th position has a nonzero value $1$, let $\hat{f}_i$ be the preimage of $e_i$, i.e., the local expansion of $\hat{f}_i$ satisfies
\[
\hat{f}_i = \tau^{-i} + \sum_{\ell=0}^{\infty} \lambda_\ell^{(i)} \tau^\ell \quad \text{for} \quad i = 1, \ldots, v.
\]
Since $\lambda_0^{(i)} \in \mathbb{F}_q \subseteq \mathcal{L}(D + vP)$, we have $\hat{f}_i - \lambda_0^{(i)} = \tau^{-i} + \sum_{\ell=1}^{\infty} \lambda_\ell^{(i)} \tau^\ell \in \mathcal{L}(D + vP)$, so we can set 
\[
\hat{f}_i = \tau^{-i} + \sum_{\ell=1}^{\infty} \lambda_\ell^{(i)} \tau^\ell \quad \text{for} \quad i = 1, \ldots, v.
\]

Let $\hat{f}_0 = 1 \in \mathbb{F}_q$. Since $\nu_P(\hat{f}_i) = -i$ for $i = 0, 1, \ldots, v$, the functions $\hat{f}_0, \hat{f}_1, \ldots, \hat{f}_v$ are linearly independent over $\mathbb{F}_q$. This implies that $\hat{f}_0, \hat{f}_1, \ldots, \hat{f}_v$ form a basis of the Riemann-Roch space $\mathcal{L}(D + vP)$ since $\ell(D + vP) = v + 1$.
\end{proof}

\subsection{The General Framework of code-based DMM}\label{sec:2.4}

In this subsection, we introduce the general framework of code-based distributed matrix multiplication and explain the concept of the recovery threshold. We assume \(A \in \mathbb{F}_q^{t \times r}\) and \(B \in \mathbb{F}_q^{r \times s}\) are two large matrices defined over \(\mathbb{F}_q\).

\begin{Fra}{[RS-based DMM]}
	
\begin{itemize}
	\item \textbf{Master node encoding}
	\begin{enumerate}
		\item \textbf{Splitting of input matrices}: Assume \(m\), \(n\), and \(p\) are divisors of \(t\), \(s\), and \(r\), respectively. The matrices \(A\) and \(B\) are partitioned into \(mp\) and \(np\) submatrices:
		\[
		A =
		\begin{pmatrix}
			A_{11} & \cdots & A_{1p} \\
			\vdots & \ddots & \vdots \\
			A_{m1} & \cdots & A_{mp}
		\end{pmatrix}
		, \quad
		B =
		\begin{pmatrix}
			B_{11} & \cdots & B_{1n} \\
			\vdots & \ddots & \vdots \\
			B_{p1} & \cdots & B_{pn}
		\end{pmatrix}
		,
		\]
		Then,
		\[
		AB =
		\begin{pmatrix}
			C_{11} & \cdots & C_{1n} \\
			\vdots & \ddots & \vdots \\
			C_{m1} & \cdots & C_{mn}
		\end{pmatrix}
		,
		\]
		where \(C_{iw} = \sum_{j=1}^p A_{ij} B_{jw}\).
		
	\item \textbf{Encoding}: Let \(\alpha_1, \alpha_2, \dots, \alpha_N\) be \(N\) distinct points in \(\mathbb{F}_q\). Define \(f(x) = \sum_{i,j} A_{i,j} x^{a_{i,j}}\) and \(g(x) = \sum_{k,w} B_{k,w} x^{b_{k,w}}\) as two polynomials, where the coefficient of the product \(h(x) = f(x)g(x)\) includes all entries \(\{C_{i,w} \mid 1 \leq i \leq m, 1 \leq w \leq n\}\). For \(1 \leq i \leq N\), the master node sends the evaluations  \(f(\alpha_i)\) and \(g(\alpha_i)\) to the \(i\)-th worker node.

	\end{enumerate}

	\item \textbf{Worker nodes local computation}: For \(1 \leq i \leq N\), the \(i\)-th worker node computes the matrix multiplication of \(f(\alpha_i)\) and \(g(\alpha_i)\) to obtain the value of \(h(\alpha_i)\).

	\item \textbf{Master node decoding}: Let \(R \) be the least number of evaluations of $h$ to uniquely determine $h$.  The master node downloads the results from \(R\) successful worker nodes to interpolate \(h(x)\) and retrieve the product \(AB\).
\end{itemize}   
\end{Fra}

The recovery threshold \(R\) is defined as the minimum number of worker nodes required to recover \(AB\). In RS-based DMM, the recovery threshold \(R = \deg(h(x)) + 1\) by Lagrange's Interpolation formula.

We can naturally generalize the framework of RS-based DMM to AG-based DMM by using functions instead of polynomials. Let \( F/\mathbb{F}_q \) be an algebraic function field with the full constant field \( \mathbb{F}_q \), and let \( f_{i,j} \) and \( g_{k,w} \) be functions in \( F \):

\begin{Fra}{[AG-based DMM]}
	
\begin{itemize}
	\item \textbf{Master node encoding}
	\begin{enumerate}
		\item \textbf{Splitting of input matrices}:  Similar to RS-based DMM.
		
	\item \textbf{Encoding}: Let \(P_1, P_2, \dots, P_N\) be \(N\) distinct rational points of $F$. Define two functions \(f = \sum_{i,j} A_{i,j} f_{i,j}\) and \(g = \sum_{k,w} B_{k,w} g_{k,w}\) such that the 
coefficient of the product \(h = fg\) includes all entries \(\{C_{i,w} \mid 1 \leq i \leq m, 1 \leq w \leq n\}\). For \(1 \leq i \leq N\), the master node sends the evaluations of \(f(P_i)\) and \(g(P_i)\) to the \(i\)-th worker node.
 
	\end{enumerate}

	\item \textbf{Worker nodes local computation}: For \(1 \leq i \leq N\), the \(i\)-th worker node computes the matrix multiplication of \(f(P_i)\) and \(g(P_i)\) to obtain the value of \(h(P_i)\).

	\item \textbf{Master node decoding}: Let \(R \) be the least number of evaluations of $h$ to uniquely determine $h$. The master node downloads the results from \(R\) successful worker nodes to decode \(h\) and retrieve the product \(AB\).
\end{itemize}   
\end{Fra}

In AG-based DMM, suppose $h \in \mathcal{L}(G)$, then the recovery threshold $R= \deg(G) + 1$ by the discussion in Section \hyperref[sec:2.2]{2.2}.

\begin{Rmk}
Let \( F \) be the rational function field \( \mathbb{F}_q(x) \) and set \( f_{i,j} = x^{a_{i,j}} \), \( g_{k,w} = x^{b_{k,w}} \).
Then, the AG-based DMM is effectively the RS-based DMM. Denote the degree of \( h \) as \( c \), then \( h \in \mathbb{F}_q[x]_{\leq c} = \mathcal{L}(cP_{\infty}) \), where \( P = P_{\infty} \) is the unique pole place of \( x \). Hence, the recovery threshold is \( c + 1 \).

\end{Rmk}

Moreover, we refer to the communication cost from the master node to the \( N \) worker nodes as the upload cost, the communication cost from the \( R \) successful worker nodes back to the master node as the download cost, and the complexity for each worker node to compute the matrix multiplications as the worker computation complexity.

\section{Extending Polynomial Codes to Algebraic Function Fields} \label{sec:3}
 In this section, we first introduce Polynomial codes \cite{yu2017polynomial}, and then employ Lemma \hyperref[lem:2.2]{2.2} to extend the Polynomial codes to the AG case.

\subsection{Polynomial Codes}  \label{sec:3.1}
In the Polynomial codes presented in \cite{yu2017polynomial}, the matrices are partitioned according to the outer product partitioning. Consider two matrices $A \in \mathbb{F}_q^{t \times r}$ and $B \in \mathbb{F}_q^{r \times s}$. Splitting the matrices $A$ and $B$ into $m$ and $n$ submatrices, respectively:
\[
A =
\begin{pmatrix}
	A_1 \\
	\vdots \\
	A_m
\end{pmatrix}
, \quad B = 
\begin{pmatrix}
	B_1 & \cdots & B_n
\end{pmatrix}, \tag{3} \label{eq:3}
\]
with $A_i \in \mathbb{F}_q^{\frac{t}{m} \times r}$ and $B_i \in \mathbb{F}_q^{r \times \frac{s}{n}}$. Then their product $AB$ can be expressed as
\[
AB =
\begin{pmatrix}
	A_1 B_1 & \cdots & A_1 B_n \\
	\vdots & \ddots & \vdots \\
	A_m B_1 & \cdots & A_m B_n
\end{pmatrix}.
\]

The master node constructs two polynomials (with matrix coefficients):
\[
f(x) := \sum_{i=1}^{m} A_i x^{a_i}, \quad g(x) := \sum_{j=1}^{n} B_j x^{b_j},
\]
with
\[
h(x) = f(x)g(x) = \sum_{i=1}^{m} \sum_{j=1}^{n} A_i B_j x^{a_i + b_j}.
\]

The master node needs to recover the value of each submatrix \(A_i B_j\), which corresponds to the coefficient of the monomial \(x^{a_i + b_j}\) in \(h(x)\). Therefore, we need \(x^{a_i + b_j}\) to have distinct degrees for different pairs \((i, j)\).
In other words, we need the following condition:

\begin{Cond}\label{cond:1}
	\[
	a_i + b_j \neq a_k + b_w \quad \text{if} \quad (i, j) \neq (k, w). \tag{4} \label{eq:4}
	\]
\end{Cond}
\noindent The recovery threshold $R = \deg(h) + 1 = \max_{1\leq i\leq m,\ 1\leq j\leq n}\left\{a_i + b_j \right\} + 1$.

In the Polynomial codes presented in \cite{yu2017polynomial}, set $a_i = i - 1$ for $i = 1, 2, \ldots, m$ and $b_j = (j - 1)m$ for $j = 1, 2, \ldots, n$. We have
\[
f(x) = \sum_{i=1}^{m} A_i x^{i-1}, \quad g(x) = \sum_{j=1}^{n} B_j x^{(j-1)m},
\]
and
\[
h(x) = f(x)g(x) = \sum_{i=1}^{m} \sum_{j=1}^{n} A_i B_j x^{i-1 + (j-1)m},
\]
satisfying the condition stated in (\hyperref[eq:4]{4}). The recovery threshold $R = mn$, which is optimal since \(\{a_i + b_j \mid 1 \leq i \leq m, \ 1 \leq j \leq n\}\) has $mn$ distinct values.

\subsection{AG-based Polynomial Codes} \label{sec:3.2}

Let $F/\mathbb{F}_q$ be an algebraic function field with the full constant field $\mathbb{F}_q$. Let $D \in \text{Div}(F)$ be a positive non-special divisor with $\deg(D) = g$ and $\ell(D) = 1$, and let $P \in \mathbb{P}_F$ be a rational point such that $P \not\in \text{supp}(D)$. Let $\tau \in F$ be a local parameter at $P$, i.e., $\nu_P(\tau) = 1$.

In our AG-based Polynomial Codes, consider two matrices \( A \in \mathbb{F}_q^{t \times r} \) and \( B \in \mathbb{F}_q^{r \times s} \), which are partitioned according to the outer product partitioning similar to (\hyperref[eq:3]{3}).
The master node constructs two functions (with matrix coefficients):
\[
f := \sum_{i=1}^{m} A_i f_i, \quad g := \sum_{j=1}^{n} B_j g_j,
\]
with
\[
h = fg = \sum_{i=1}^{m} \sum_{j=1}^{n} A_i B_j f_i g_j.
\]

Now each submatrix \(A_i B_j\) is the coefficient of \(f_i g_j\). Assume \(h\) belongs to a Riemann-Roch space \(\mathcal{L}(G)\) for some divisor \(G\). Therefore, we need \(f_i g_j\) to have distinct discrete valuations at the rational place \(P\) for different pairs \((i, j)\), ensuring that the set \(\{f_i g_j \mid 1 \leq i \leq m, \ 1 \leq j \leq n\}\) can be extended to a basis of \(\mathcal{L}(G)\) as it is \(\mathbb{F}_q\)-linearly independent. Then \(AB\) can be recovered from the expression of \(h\) under this basis. In other words, we need the following conditions:

\begin{Cond}\label{cond:2}
\[
\nu_P(f_i g_j) \neq \nu_P(f_k g_w) \quad \text{if} \quad (i, j) \neq (k, w).  \tag{5} \label{eq:5}
\]
\end{Cond}

In our setup, we consider functions in the Riemann-Roch space $\mathcal{L}(D + \infty P) = \bigcup_{i=0}^{\infty} \mathcal{L}(D + i P)$, i.e., we choose $f_i, g_j \in \mathcal{L}(D + \infty P)$ for $i = 1, 2, \ldots, m$ and $j = 1, 2, \ldots, n$.

\begin{Con} \label{con:1}[AG-based Polynomial codes]	
\begin{itemize}

    \item[(1)] $n=1$. Consider the Riemann-Roch space $\mathcal{L}(D + (m-1)P)$. 
	By Lemma \hyperref[lem:2.2]{2.2}, we can construct a basis $\hat{f}_0, \ldots, \hat{f}_{m-1}$ of $\mathcal{L}(D + (m-1)P)$ satisfying $\hat{f}_0 = 1 \in \mathbb{F}_q$ and
	\[
	\hat{f}_i = \tau^{-i} + \sum_{\ell=1}^{\infty} \lambda_\ell^{(i)} \tau^\ell \quad \text{for} \quad i = 1, \ldots, m-1,
	\]
	where $\lambda_\ell^{(i)} \in \mathbb{F}_q$ for $i = 1, \ldots, m-1$.

 \item[(2)] $n \ge 2$. Consider the Riemann-Roch space $\mathcal{L}(D + (n-1)mP)$. 
	By Lemma \hyperref[lem:2.2]{2.2}, we can construct a basis $\hat{f}_0, \ldots, \hat{f}_{(n-1)m}$ of $\mathcal{L}(D + (n-1)mP)$ satisfying $\hat{f}_0 = 1 \in \mathbb{F}_q$ and
	\[
	\hat{f}_i = \tau^{-i} + \sum_{\ell=1}^{\infty} \lambda_\ell^{(i)} \tau^\ell \quad \text{for} \quad i = 1, \ldots, (n-1)m,
	\]
	where $\lambda_\ell^{(i)} \in \mathbb{F}_q$ for $i = 1, \ldots, (n-1)m$.
	
	Hence, set $f_1 = g_1 = \hat{f}_0 \in \mathbb{F}_q$, and
	\[
	f_i = \hat{f}_{i-1} = \tau^{-(i-1)} + \sum_{\ell=1}^{\infty} \lambda_\ell^{(i-1)} \tau^\ell \quad \text{for} \quad i = 2, \ldots, m,
	\]
	\[
	g_j = \hat{f}_{(j-1)m} = \tau^{-(j-1)m} + \sum_{\ell=1}^{\infty} \lambda_\ell^{(j-1)m} \tau^\ell \quad \text{for} \quad j = 2, \ldots, n.
	\]
\end{itemize}
	
\end{Con}

For Construction \hyperref[con:1]{1}, we have the following theorem:

\begin{Thm}\label{thm:3.1}
For any prime power \( q \) and partitioning parameters \( m, n \ge 1 \) that divide \( t \) and \( s \), respectively, Construction \hyperref[con:1]{1} provides an AG-based DMM scheme in \( \mathbb{F}_q \) with a recovery threshold of \( R = 2g + mn \).

\end{Thm}

\begin{proof}
First, we prove that our construction satisfies Condition \hyperref[cond:2]{2}. In our setup, we have
\[
f_i g_j = \tau^{-(i-1 + (j-1)m)} + \sum_{\ell=v}^{\infty} \lambda_{\ell}^{(i,j)} \tau^{\ell},
\]
where $v = -\max\left\lbrace (i-1), (j-1)m \right\rbrace + 1$.
We can get $\nu_P(f_i g_j) = -((i-1) + (j-1)m)$ for $i = 1, \ldots, m$ and $j = 1, \ldots, n$. Thus, $\nu_P(f_i g_j)$ is distinct for different pairs $(i, j)$, which satisfies the condition stated in (\hyperref[eq:5]{5}).

Next, we consider the recovery threshold. Since $$\min \left\lbrace \nu_P(f_i g_j) \right\rbrace = -((m-1) + (n-1)m) = -(mn - 1)$$ we have $f_i g_j \in \mathcal{L}(2D + (mn-1)P)$ for all $i = 1, \ldots, m$ and $j = 1, \ldots, n$. Hence, $h = fg \in \mathcal{L}(2D + (mn-1)P)^{\frac{t}{m} \times \frac{s}{n}}$. Denote the $(i,j)$-th entry of $h$ as $h^{(i,j)}$, then for each $(i,j) \in \left[1, \frac{t}{m}\right] \times \left[1, \frac{s}{n}\right]$, $h^{(i,j)} \in \mathcal{L}(2D + (mn-1)P)$. Therefore, the recovery threshold is $R = \deg(2D + (mn-1)P) + 1 = 2g + mn$.
\end{proof}

Also, if $m$ or $n$ satisfies certain special conditions, we can further reduce the recovery threshold in Construction \hyperref[con:1]{1}. Without loss of generality, we discuss the parameter $m$. Observe that if $m \in W(P)$, where $W(P)$ is the Weierstrass semigroup of $P$, then $2m, \ldots, (n-1)m$ are all in $W(P)$. In this case, we choose $g_j$ in the one-point Riemann-Roch space:

\begin{Con} \label{con:2}[AG-based Polynomial codes for special $m$]	

\noindent Assume $m \in W(P)$, i.e., $m$ is a pole number of $P$. 
Set $f_i$ the same as in Construction \hyperref[con:1]{1} for $i = 1, 2, \ldots, m$ and set $g_1 = 1 \in \mathbb{F}_q$ and $g_j \in \mathcal{L}((j-1)mP) \backslash \mathcal{L}(((j-1)m-1)P)$ for $j = 2, \ldots, n$, i.e.,
\[
g_j = \tau^{-(j-1)m} + \sum_{\ell=-(j-1)m + 1}^{\infty} \hat{\lambda}_\ell^{(j)} \tau^\ell \quad \text{for} \quad j = 2, \ldots, n.
\]

\end{Con}

For Construction \hyperref[con:2]{2}, we have the following theorem:

\begin{Thm}\label{thm:3.2}
For any prime power $q$ and partitioning parameters $m, n \ge 1$ that divide $t$ and $s$ respectively, if $m \in W(P)$ or $n \in W(P)$, where $W(P)$ is the Weierstrass semigroup of $P$, Construction \hyperref[con:2]{2} gives an AG-based DMM scheme in $\mathbb{F}_q$ with a recovery threshold of $R = g + mn$.

\end{Thm}
\begin{proof}
The proof is similar to the proof of Theorem \hyperref[thm:3.1]{3.1}. We can easily verify that our construction satisfies Condition \hyperref[cond:2]{2} and \( f_i g_j \in \mathcal{L}(D + (mn-1)P) \) for all \( i = 1, \ldots, m \) and \( j = 1, \ldots, n \). Hence, the recovery threshold is \( R = \deg(D + (mn-1)P) + 1 = g + mn \).
\end{proof}

\begin{Rmk}\label{rmk:3.3}
For \( m \not\in W(P) \) and \( n \not\in W(P) \), we can also reduce the recovery threshold of Construction \hyperref[con:1]{1} by a special choice of \( g_j \) in the one-point Riemann-Roch space, similar to Construction \hyperref[con:2]{2}. A detailed discussion will be provided in Section \hyperref[sec:7]{7}.

\end{Rmk}

%Recall that there are a total of \(N\) worker nodes and at least \(R\) successful worker nodes. 
Next, we discuss the communication cost and per-worker computation complexity. Let \(N\) be the total number of worker nodes and \(R\) be the recovery threshold.

\begin{itemize}
    \item For communication cost, the master node sends \(O\left(\frac{tr}{m} + \frac{rs}{n}\right)\) symbols to each worker node and receives \(O\left(\frac{ts}{mn}\right)\) symbols from each successful worker node. Therefore, the upload cost is \(O\left(\left(\frac{t}{m} + \frac{s}{n}\right)rN\right)\), and the download cost is \(O\left(\frac{ts}{mn}R\right)\).
    \item For worker computation complexity, since each worker node computes matrix multiplications of size \(\frac{t}{m} \times r\) and \(r \times \frac{s}{n}\), the worker complexity is \(O\left(\frac{trs}{mn}\right)\) using the native multiplication algorithm.
\end{itemize}

	\section{Extending MatDot Codes to Algebraic Function Fields}\label{sec:4}
	In this section, we first introduce MatDot codes \cite{dutta2019optimal}, and then employ Lemma \hyperref[lem:2.2]{2.2} to extend the MatDot codes to the AG case.

	\subsection{MatDot Codes}\label{sec:4.1}
	%Compared to the Polynomial codes described in Section \hyperref[sec:3.1]{3.1}, MatDot codes in \cite{dutta2019optimal} have a better recovery threshold $R$ but worse per-worker computation and communication costs. 

In the MatDot codes presented in \cite{dutta2019optimal}, the matrices are partitioned according to the inner product partitioning. Consider two matrices $A \in \mathbb{F}_q^{t \times r}$ and $B \in \mathbb{F}_q^{r \times s}$. Splitting the matrices $A$ and $B$ both into $p$ submatrices:
\[
A = \begin{pmatrix}
	A_1 & \cdots & A_p
\end{pmatrix}
, \quad B =
\begin{pmatrix}
	B_1 \\
	\vdots \\
	B_p
\end{pmatrix}
, \tag{6} \label{eq:6}
\]
with $A_i \in \mathbb{F}_{q}^{t \times \frac{r}{p}}$ and $B_i \in \mathbb{F}_{q}^{\frac{r}{p} \times s}$. Then their product $AB$ can be expressed as:
\[
AB = \sum_{i=1}^p A_i B_i.
\]

The master node constructs two polynomials (with matrix coefficients):
\[
f(x) := \sum_{i=1}^{p} A_i x^{a_i}, \quad g(x) := \sum_{j=1}^{p} B_j x^{b_j},
\]
with
\[
h(x) = f(x)g(x) = \sum_{i=1}^{p} \sum_{j=1}^{p} A_i B_j x^{a_i + b_j}.
\]

The master node needs to recover the value of each submatrix \( A_i B_i \), which is the coefficient of the monomial \( x^{a_i + b_i} \) in \( h(x) \). Hence, we need \( x^{a_i + b_i} \) to have the same degree \( d \) for \( i = 1, 2, \ldots, p \), and \( x^{a_i + b_j} \) to have a degree different from \( d \) for \( i \neq j \). In other words, we need the following condition:

\begin{Cond}\label{cond:3} Let $d$ be an integer. Then
	\[
	a_i + b_j = d \quad \text{if and only if} \quad i=j. \tag{7} \label{eq:7}
	\]
\end{Cond}
\noindent  The recovery threshold $R = \deg(h) +1 = \text{max}_{1\leq i,j\leq p}\left\lbrace a_i + b_j \right\rbrace +1$.

In the Matdot codes presented in \cite{dutta2019optimal}, set $a_i = i - 1$ for $i = 1, 2, \ldots, p$ and $b_j = p - j$ for $j = 1, 2, \ldots, p$, so $d = p - 1$. We have
\[
f(x) = \sum_{i=1}^{p} A_i x^{i-1}, \quad g(x) = \sum_{j=1}^{p} B_j x^{p-j},
\]
and
\[
h(x) = f(x)g(x) = \sum_{i=1}^{p} \sum_{j=1}^{p} A_i B_j x^{i-j + p - 1},
\]
satisfying the condition stated in (\hyperref[eq:7]{7}). The recovery threshold is $R = 2p - 1$.

\subsection{AG-based MatDot Codes}
We continue the notations used in Subsection~\ref{sec:3.2}.
%Let $F/\mathbb{F}_q$ be an algebraic function field with the full constant field $\mathbb{F}_q$. Let $D \in \text{Div}(F)$ be a positive non-special divisor with $\deg(D) = g$ and $\ell(D) = 1$, and let $P \in \mathbb{P}_F$ be a rational point such that $P \not\in \text{supp}(D)$. Let $t \in F$ be a local parameter at $P$, i.e., $\nu_P(t) = 1$.
In our AG-based MatDot Codes, consider two matrices \( A \in \mathbb{F}_q^{t \times r} \) and \( B \in \mathbb{F}_q^{r \times s} \), which are partitioned according to the inner product partitioning similar to (\hyperref[eq:6]{6}).
The master node constructs two functions (with matrix coefficients):
\[
f := \sum_{i=1}^{p} A_i f_i, \quad g := \sum_{j=1}^{p} B_j g_j,
\]
with
\[
h = fg = \sum_{i=1}^{p} \sum_{j=1}^{p} A_i B_j f_i g_j.
\]
Now, each submatrix \(A_i B_i\) is the coefficient of \(f_i g_i\). Therefore, we need \(f_i g_i\) to have a same discrete valuation \(d\) at the point \(P\) for \(i = 1, 2, \ldots, p\), and \(\nu_P(f_i g_j)\neq d\) for \(i \neq j\). Additionally, we require that any $\F_q$-linear combination of \(f_i g_j\) with \(i \neq j\) has discrete valuations not equal to \(d\), i.e, $\nu_P\left(\sum_{i\neq j} a_{i,j} f_i g_j\right) \neq d$. This ensures that $\sum_{i\neq j} A_i B_j f_i g_j$ do not affect the coefficients of \(\{f_i g_i\}_{1\leq i\leq p}\) in $h$, avoiding the cancellation problem. In other words, we need the following condition:
\begin{Cond}\label{cond:4} Let $d$ be an integer. Then
\[
\nu_P(f_i g_j) = d \quad \text{if and only if} \quad i = j, \tag{8} \label{eq:8}
\]
\[
\nu_P\left(\sum a_{i,j} f_i g_j\right) \neq d \quad \text{for} \quad \forall\ i \neq j \; \text{and} \; \; \forall\ a_{i,j} \in \mathbb{F}_q. \tag{9} \label{eq:9}
\]

\end{Cond}

\begin{Rmk} \label{rmk:4.1}
	
In fact, we only need to consider \(f_i g_j\) with discrete valuations less than \(d\), since terms \(f_i g_j\) with discrete valuations greater than \(d\) will not affect the coefficient of \(f_i g_i\) with discrete valuations equal to \(d\). Therefore, the second condition (\hyperref[eq:9]{9}) can be changed to:

\[
\nu_P\left(\sum a_{i,j} f_i g_j\right) \neq d \quad \text{for} \quad \forall\ f_i, g_j \; \text{s.t.} \; \nu_P(f_i g_j) < d \; \text{and} \; \; \forall\ a_{i,j} \in \mathbb{F}_q. \tag{10} \label{eq:10}
\]

\end{Rmk}

In our setup, we consider functions in the Riemann-Roch space $\mathcal{L}(D + \infty P) = \bigcup_{i=0}^{\infty} \mathcal{L}(D + i P)$, i.e., we choose $f_i, g_j \in \mathcal{L}(D + \infty P)$ for $i = 1, 2, \ldots, p$ and $j = 1, 2, \ldots, p$.

\begin{Con} \label{con:3} [AG-based MatDot codes]

 Consider the Riemann-Roch space $\mathcal{L}(D + (p-1)P)$. By Lemma \hyperref[lem:2.2]{2.2}, we can construct a basis $\hat{f}_0, \ldots,\hat{f}_{p-1}$ of $\mathcal{L}(D + (p-1)P)$ satisfying $\hat{f}_0 = 1 \in \mathbb{F}_q$ and
	\[
	\hat{f}_i = \tau^{-i} + \sum_{\ell=1}^{\infty} \lambda_\ell^{(i)} \tau^\ell \quad \text{for} \quad i = 1, \ldots, p-1,
	\]
	where $\lambda_\ell^{(i)} \in \mathbb{F}_q$ for $i = 1, \ldots, p-1$. 	
	
	Hence, set $f_1 = g_p = \hat{f}_0 \in \mathbb{F}_q$, and
\[
f_i = \hat{f}_{i-1} = \tau^{-(i-1)} + \sum_{\ell=1}^{\infty} \lambda_{\ell}^{(i-1)} \tau^{\ell} \quad \text{for} \quad i=2, \ldots, p,
\]

\[
g_j = \hat{f}_{p - j} = \tau^{-(p - j)} + \sum_{\ell=1}^{\infty} \lambda_{\ell}^{(p - j)} \tau^{\ell} \quad \text{for} \quad j=1, \ldots, p-1.
\]

\end{Con}

For Construction \hyperref[con:3]{3}, we have the following theorem:
\begin{Thm}\label{thm:4.2}
For any prime power $q$ and partitioning parameter $p \ge 1$ that divides $r$, Construction \hyperref[con:3]{3} gives an AG-based DMM scheme in $\mathbb{F}_q$ with a recovery threshold of $R = 2g + 2p - 1$.

\end{Thm}

\begin{proof}
First, we need to prove that our construction satisfies Condition \hyperref[cond:4]{4}. In this case, we have $f_i = g_{p + 1 - i}$ for $i = 1, \ldots, p$ and
\[
f_i g_j = \tau^{-(p-1 + i - j)} + \sum_{\ell=-(p-2)}^{\infty} \lambda_{\ell}^{(i,j)} \tau^{\ell} \quad \text{for} \quad i \ge j.
\]
We can get $\nu_P(f_i g_i) = -(p-1)$ for $i = 1, \ldots, p$ and $\nu_P\left(\sum a_{i,j} f_i g_j\right) \neq -(p-1)$ for $\forall i > j$ and $\forall a_{i,j} \in \mathbb{F}_q$, which satisfies the condition stated in (\hyperref[eq:8]{8}) and (\hyperref[eq:10]{10}).

Next, we consider the recovery threshold. Since $\min \left\lbrace \nu_P(f_i g_j) \right\rbrace = -(2(p-1))$, we have $f_i g_j \in \mathcal{L}(2D + 2(p-1)P)$ for all $i = 1, \ldots, p$ and $j = 1, \ldots, p$. Hence, $h = fg \in \mathcal{L}(2D + 2(p-1)P)^{t \times s}$. Denote the $(i,j)$-th entry of $h$ as $h^{(i,j)}$, then for each $(i,j) \in [1,t] \times [1,s]$, $h^{(i,j)} \in \mathcal{L}(2D + 2(p-1)P)$. Therefore, the recovery threshold is $R = \deg(2D + 2(p-1)P) + 1 = 2g + 2p - 1$.
\end{proof}

In our setup, we can recover \( AB \) from the local expansion of \( h \) at the point \( P \).
\begin{Rmk}Since
\[
f_i g_j = \tau^{-(p-1 + i - j)} + \sum_{\ell=v}^{\infty} \lambda_{\ell}^{(i,j)} \tau^{\ell},
\]
where \( v = - \max \{ i-1, p-j \} + 1 \ge -(p-2) \), the product of \( f_i \) and \( g_j \) will not affect the coefficients of the term \(\tau^{-(p-1)}\) in the local expansion of \( h \). Hence, we can recover \( AB \) from the coefficient of \(\tau^{-(p-1)}\) in the local expansion of \( h \) at the point \( P \). This coefficient should correctly represent the values of \(\sum_{j=1}^{p} A_{ij} B_{j}\).

\end{Rmk}

Next, we discuss the communication cost and per-worker computation complexity. Let \(N\) be the total number of worker nodes and \(R\) be the recovery threshold.

\begin{itemize}
     \item For communication cost, the master node sends \(O\left(\frac{tr}{p} + \frac{rs}{p}\right)\) symbols to each worker node and receives \(O\left(ts\right)\) symbols from each successful worker node. Therefore, the upload cost is \(O\left(\left(\frac{tr}{p} + \frac{rs}{p}\right)N\right)\), and the download cost is \(O\left(tsR\right)\).

    \item For worker computation complexity, since each worker node computes matrix multiplications of size \(t \times \frac{r}{p}\) and \(\frac{r}{p} \times s\), the worker computation complexity is \(O\left(\frac{trs}{p}\right)\) using the native multiplication algorithm.
\end{itemize}

\section{Extending PolyDot Codes to Algebraic Function Fields}\label{sec:5}
In this section, we first introduce classical PolyDot codes \cite{dutta2019optimal}, Entangled Polynomial codes \cite{yu2020straggler}, and Generalized PolyDot codes \cite{dutta2018unified}. We then propose a new polynomial-based construction of PolyDot code that achieves the same recovery threshold as PolyDot codes. At last, we extend our new PolyDot code to algebraic function field.

	\subsection{PolyDot Codes}\label{sec:5.1}
%	PolyDot codes aim to bridge the gap between Polynomial codes and MatDot codes, offering intermediate per-worker computation/communication costs and recovery thresholds. Polynomial codes \cite{yu2017polynomial} and MatDot codes \cite{dutta2019optimal} are two special cases of PolyDot codes.

PolyDot codes were introduced in \cite{dutta2019optimal}. Meanwhile, \cite{yu2020straggler} presented Entangled Polynomial code, and \cite{dutta2018unified} introduced Generalized PolyDot codes, both of which have smaller recovery thresholds. Assume  $A \in \mathbb{F}_q^{t \times r}$ and $B \in \mathbb{F}_q^{r \times s}$. In all these cases, $A, B$ are partitioned as follows
\[
A =
\begin{pmatrix}
	A_{11} & \cdots & A_{1p} \\
	\vdots & \ddots & \vdots \\
	A_{m1} & \cdots & A_{mp}
\end{pmatrix}
, \quad
B =
\begin{pmatrix}
	B_{11} & \cdots & B_{1n} \\
	\vdots & \ddots & \vdots \\
	B_{p1} & \cdots & B_{pn}
\end{pmatrix}
, \tag{11} \label{eq:11}
\]
with $A_{ij} \in \mathbb{F}_q^{\frac{t}{m} \times \frac{r}{p}}$ and $B_{kw} \in \mathbb{F}_q^{\frac{r}{p} \times \frac{s}{n}}$. Their product $AB$ can be expressed as
\[
AB =
\begin{pmatrix}
	C_{11} & \cdots & C_{1n} \\
	\vdots & \ddots & \vdots \\
	C_{m1} & \cdots & C_{mn}
\end{pmatrix}
,
\]
where $C_{iw} = \sum_{j=1}^p A_{ij} B_{jw}$.

The master node constructs two polynomials (with matrix coefficients):
\[
f(x) := \sum_{i=1}^{m} \sum_{j=1}^{p}  A_{ij}x^{(i-1)\alpha +(j-1)\beta}, \quad g(x) := \sum_{k=1}^{p} \sum_{w=1}^{n}B_{kw} x^{(p-k)\beta + (w-1) \theta},
\]
with
\[
h(x) = f(x)g(x) =  \sum_{i=1}^{m} \sum_{j=1}^{p} \sum_{k=1}^{p} \sum_{w=1}^{n}A_{ij} B_{kw} x^{(i-1)\alpha + (p-1 -k+j)\beta + (w-1) \theta}.
\]
where $\alpha,\beta,\theta \in \NN$.

\begin{Rmk} \label{rmk:5.1}
	Note that the degree selection in $x^{(i-1)\alpha +(j-1)\beta}$ and $x^{(p-k)\beta + (w-1) \theta}$ relies on the idea from MatDot codes. In this case, for each $(i, w) \in [1,m] \times [1,n]$, the $(i,w)$-th submatrix $C_{i,w}=\sum_{j=1}^{p}A_{ij} B_{jw}$ of $AB$ is exactly the coefficient of the monomial $x^{(i-1)\alpha + (p-1)\beta + (w-1) \theta}$ in $h(x)$.

\end{Rmk}

For each \((i, w) \in [1, m] \times [1, n]\), let \( d_{i,w} = (i-1)\alpha + (p-1)\beta + (w-1)\theta \). We require that \( x^{(i'-1)\alpha + (p-1 - k' + j')\beta + (w'-1)\theta} \) has the degree \( d_{i,w} \) if and only if \( i' = i \), \( j' = k' \), and \( w' = w \) for any \((i', j', k', w') \in [1, m] \times [1, p] \times [1, p] \times [1, n] \). Then \( AB \) can be recovered from the coefficients of \(\{ x^{d_{i,w}} \}_{(i, w) \in [1, m] \times [1, n]} \) in \( h(x) \). In other words, we need the following condition:

\begin{Cond}\label{cond:5}
Let $d_{i,w}=(i-1)\alpha + (p-1)\beta + (w-1)\theta$ for each $(i,w) \in [1,m] \times [1,n]$, where $\alpha,\beta,\theta \in \NN$. Then,
	\[
	i'\alpha + (j'-k')\beta + w' = d_{i,w}\quad \text{if and only if} \quad (i', j'-k',w')= (i, 0,w). \tag{12} \label{eq:12}
	\]
\end{Cond}

\noindent The recovery threshold $R = \deg(h) +1 = (m-1)\alpha + 2(p-1)\beta + (n-1)\theta +1$.

In PolyDot codes presented in \cite{dutta2019optimal}, set $\alpha = 1$, $\beta = m$, and $\theta = m(2p-1)$. We have
\[
f(x) = \sum_{i=1}^{m} \sum_{j=1}^{p} A_{ij}x^{(i-1) + (j-1)m}, \quad g(x) = \sum_{k=1}^{p} \sum_{w=1}^{n} B_{kw} x^{(p-k)m + (w-1)m(2p-1)},
\]
and
\[
h(x) = f(x)g(x) = \sum_{i=1}^{m} \sum_{j=1}^{p} \sum_{k=1}^{p} \sum_{w=1}^{n} A_{ij} B_{kw} x^{(i-1) + (p-1 - k + j)m + (w-1)m(2p-1)},
\]
satisfying the condition stated in (\hyperref[eq:12]{12}). The recovery threshold is $R = (2p-1)mn$.

In Entangled Polynomial codes presented in \cite{yu2020straggler} and  Generalized PolyDot codes presented in \cite{dutta2018unified}, set $\alpha = p$, $\beta = 1$, and $\theta = mp$. We have
\[
f(x) = \sum_{i=1}^{m} \sum_{j=1}^{p} A_{ij} x^{(i-1)p + (j-1)}, \quad g(x) = \sum_{k=1}^{p} \sum_{w=1}^{n} B_{kw} x^{p-k + (w-1)mp},
\]
and
\[
h(x) = f(x)g(x) = \sum_{i=1}^{m} \sum_{j=1}^{p} \sum_{k=1}^{p} \sum_{w=1}^{n} A_{ij} B_{kw} x^{(i-1)p + p-1 - k + j + (w-1)mp},
\]
satisfying the condition stated in (\hyperref[eq:12]{12}). The recovery threshold is $R = pmn + p - 1$.

To extend PolyDot Codes to the AG-based case, we consider the following new construction:
Set \(\alpha = 1\), \(\beta = mn\), and \(\theta = m\).
We have
	\[
	f(x) = \sum_{i=1}^{m} \sum_{j=1}^{p} A_{ij} x^{(i-1) + (j-1)mn}, \quad g(x) = \sum_{k=1}^{p} \sum_{w=1}^{n} B_{kw} x^{mn(p-k) + (w-1)m},
	\]
	and
	\[
	h(x) = f(x)g(x) = \sum_{i=1}^{m} \sum_{j=1}^{p} \sum_{k=1}^{p} \sum_{w=1}^{n} A_{ij} B_{kw} x^{(i-1) + mn(p-1 - k + j) + (w-1)m},
	\]
	satisfying the condition stated in (\hyperref[eq:12]{12}). The recovery threshold is $R = (2p-1)mn$.

\subsection{AG-based PolyDot Codes}
Let $F/\mathbb{F}_q$ be an algebraic function field with the full constant field $\mathbb{F}_q$. Let $P$ and $Q$ be two distinct rational points in $\mathbb{P}_F$. Let $\tau \in F$ be a local parameter at $P$, i.e., $\nu_P(\tau) = 1$.

In our AG-based PolyDot Codes, consider two matrices \( A \in \mathbb{F}_q^{t \times r} \) and \( B \in \mathbb{F}_q^{r \times s} \), which are partitioned according to the general product partitioning similar to (\hyperref[eq:11]{11}).
The master node constructs two functions (with matrix coefficients):
\[
f(x) := \sum_{i=1}^{m} \sum_{j=1}^{p}  A_{ij}f_{i,j}, \quad g(x) := \sum_{k=1}^{p} \sum_{w=1}^{n}B_{kw} g_{k,w},
\]
with
\[
h = fg =  \sum_{i=1}^{m} \sum_{j=1}^{p} \sum_{k=1}^{p} \sum_{w=1}^{n}A_{ij} B_{kw} f_{i,j}g_{k,w}.
\]

Now, each submatrix $\sum_{j=1}^{p} A_{ij} B_{jw}$ is the coefficient of $f_{i,j} g_{j,w}$. Similar to the case in PolyDot codes, we need to satisfy the following two requirements: (1) Fix $(i,w) \in [1,m] \times [1,n]$, we need $f_{i,j} g_{j,w}$ to have the same discrete valuation $d_{i,w}$ at the point $P$ for $j = 1, 2, \ldots, p$; (2) For any $(i', j', k', w') \in [1,m] \times [1,p] \times [1,p] \times [1,n]$, the discrete valuation of $f_{i',j'} g_{k',w'}$ at the point $P$ is $d_{i,w}$ only if $i' = i$, $j' = k'$, and $w' = w$. In conclusion, we need the following condition:
\begin{Cond}\label{cond:6}
Let $d_{i,w}$ be distinct integers for each $(i,w) \in [1,m] \times [1,n]$. Then,
\[
\nu_P(f_{i',j'} g_{k',w'}) = d_{i,w} \quad \text{if and only if} \quad (i', j' - k', w')=(i, 0, w) . \tag{13} \label{eq:13}
\]

\end{Cond}

\begin{Rmk} \label{rmk:5.2}
Moreover, as discussed in Section \hyperref[sec:4]{4}, for each $(i, w) \in [1,m] \times [1,n]$, we require that those $f_{i',j'} g_{k',w'}$ with discrete valuations less than $d_{i,w}$ have discrete valuations not equal to $d_{i,w}$ after any $\mathbb{F}_q$-linear combination to avoid the cancellation problem . However, by using local expansions in our construction, we can ensure that the above condition is always satisfied.

\end{Rmk}

Unlike in Sections \hyperref[sec:3]{3} and \hyperref[sec:4]{4}, here we no longer use the Riemann-Roch space constructed by a non-special divisor \(D\), but instead consider the following space:
\[
\mathcal{L}((2g + a)Q + (a + b)P),
\]
where \(a, b \ge 0\). We have the following lemma:

\begin{Lem} \label{lem:5.3}
Let $P$ and $Q$ be two distinct rational points in $\mathbb{P}_F$. Let $\tau \in F$ be a local parameter at $P$, i.e., $\nu_P(\tau) = 1$. Then, for each $a, b \ge 0$, we can construct $\hat{f}_0, \hat{f}_1, \ldots, \hat{f}_{a+b} \in \mathcal{L}((2g + a)Q + (a+b)P)$ satisfying 
\[
\hat{f}_i = \tau^{-i} + \sum_{\ell=a+1}^{\infty} \lambda_\ell^{(i)} \tau^\ell \quad \text{for} \quad i = 0, \ldots, a+b,
\]
where $\lambda_\ell^{(i)} \in \mathbb{F}_q$ for $i = 0, \ldots, a+b$. Observe that $\hat{f}_i \in \mathcal{L}((2g + a)Q + iP) \setminus \mathcal{L}((2g + a)Q + (i-1)P)$ for $i = 0, \ldots, a+b$.

\end{Lem}

\begin{proof}

For any function $\hat{f} \in \mathcal{L}((2g + a)Q + (a+b)P)$, we have $\nu_P(\hat{f}) \ge -(a+b)$. By Equation~\eqref{eq:2}, assume the local expansion of $\hat{f}$ at $P$ is
\[
\hat{f} = \sum_{\ell=-(a+b)}^{\infty} \lambda_\ell \tau^\ell, \quad \lambda_\ell \in \mathbb{F}_q.
\]

Consider the following mapping $\phi: \mathcal{L}((2g + a)Q + (a+b)P) \to \mathbb{F}_q^{2a+b+1}$ given by 
\[
\hat{f} = \sum_{\ell=-(a+b)}^{\infty} \lambda_\ell \tau^\ell \longmapsto (\lambda_0, \lambda_{-1}, \ldots, \lambda_{-(a+b)}, \lambda_{1}, \ldots, \lambda_{a}).
\]

Observe that $\phi$ is $\mathbb{F}_q$-linear with kernel $\ker(\phi) = \mathcal{L}((2g + a)Q - (1 + a)P)$. Since $\ell((2g + a)Q + (a+b)P) - \ell((2g + a)Q - (1 +a)P) = 2a + b + 1$, the mapping $\phi$ is a surjection. 
For each \( i \) satisfying \( 0 \le i \le a+b \), consider the vector \( e_i = (0, \ldots, 1, \ldots, 0) \in \mathbb{F}_q^{2a + b + 1} \), where only the \( i+1 \)-th position has a nonzero value of 1. Let \(\hat{f}_i\) be the preimage of \( e_i \), i.e., the local expansion of \(\hat{f}_i\) satisfies
$$\hat{f}_i = \tau^{-i} + \sum_{\ell=a+1}^{\infty} \lambda_\ell^{(i)} \tau^\ell \quad \text{for} \quad i = 0, \ldots, a+b.$$
\end{proof}

 In the following Construction \hyperref[con:4]{4}, we extend our new construction for PolyDot codes in Section \hyperref[sec:5.1]{5.1} to the AG case. Depending on the value of $mn$, we categorize it into the following two cases: (1) $m = 1$ or $m \ge n \ge 2$; (2) $n = 1$ or $n > m \ge 2$.

\begin{Con} \label{con:4}[AG-based PolyDot codes]
\begin{itemize}

    \item[(1)]	$m = 1$ or $m \ge n \ge 2$. We set $a = m(n-1)$ and $b = (p-1)mn$. Consider the Riemann-Roch space $\mathcal{L}((2g + m(n-1))Q + ((p-1)mn + m(n-1))P)$. By Lemma \hyperref[lem:5.3]{5.3}, we can construct $\hat{f}_0, \hat{f}_1, \ldots, \hat{f}_{(p-1)mn + m(n-1)} \in \mathcal{L}((2g + m(n-1))Q + ((p-1)mn + m(n-1))P)$ satisfying 
\[
\hat{f}_i = \tau^{-i} + \sum_{\ell=m(n-1)+1}^{\infty} \lambda_\ell^{(i)} \tau^\ell \quad \text{for} \quad i = 0, \ldots, (p-1)mn + m(n-1),
\]
where $\lambda_\ell^{(i)} \in \mathbb{F}_q$ for $i = 0, \ldots, (p-1)mn + m(n-1)$. Hence, set 
\[
f_{i,j} = \hat{f}_{(i-1) + (j-1)mn} = \tau^{-((i-1) + (j-1)mn)} + \sum_{\ell=m(n-1)+1}^{\infty} \lambda_\ell^{((i-1) + (j-1)mn)} \tau^\ell
\]
for $(i,j) \in [1,m] \times [1,p]$, and set
\[
g_{k,w} = \hat{f}_{(p-k)mn + m(w-1)} = \tau^{-((p-k)mn + m(w-1))} + \sum_{\ell=m(n-1)+1}^{\infty} \lambda_\ell^{((p-k)mn + m(w-1))} \tau^\ell
\]
for $(k,w) \in [1,p] \times [1,n]$.

 \item[(2)] $n = 1$ or $n > m \ge 2$. We set $a = n(m-1)$ and $b = (p-1)mn$. Consider the Riemann-Roch space $\mathcal{L}((2g + n(m-1))Q + ((p-1)mn + n(m-1))P)$. By Lemma \hyperref[lem:5.3]{5.3}, we can construct $\hat{f}_0, \hat{f}_1, \ldots, \hat{f}_{(p-1)mn + n(m-1)} \in \mathcal{L}((2g + n(m-1))Q + ((p-1)mn + n(m-1))P)$ satisfying 
\[
\hat{f}_i = \tau^{-i} + \sum_{\ell=n(m-1)+1}^{\infty} \lambda_\ell^{(i)} \tau^\ell \quad \text{for} \quad i = 0, \ldots, (p-1)mn + n(m-1),
\]
where $\lambda_\ell^{(i)} \in \mathbb{F}_q$ for $i = 0, \ldots, (p-1)mn + n(m-1)$. Hence, set 
\[
f_{i,j} = \hat{f}_{(i-1)n + (j-1)mn} = \tau^{-((i-1)n + (j-1)mn)} + \sum_{\ell=n(m-1)+1}^{\infty} \lambda_\ell^{((i-1)n + (j-1)mn)} \tau^\ell
\]
for $(i,j) \in [1,m] \times [1,p]$, and set
\[
g_{k,w} = \hat{f}_{(p-k)mn + (w-1)} = \tau^{-((p-k)mn + (w-1))} + \sum_{\ell=n(m-1)+1}^{\infty} \lambda_\ell^{((p-k)mn + (w-1))} \tau^\ell
\]
for $(k,w) \in [1,p] \times [1,n]$.

\end{itemize}
\end{Con}

For Construction \hyperref[con:4]{4}, we have the following theorem:

\begin{Thm}\label{thm:5.4}
For any prime power $q$ and partitioning parameters $m, p, n \ge 1$ that divide $t, r, s$, respectively, Construction \hyperref[con:4]{4} gives an AG-based PoltDot code over $\mathbb{F}_q$ with a recovery threshold of 
\[
R = \left\{
\begin{array}{ll}
	4g + (2p-1)mn + 2mn - 2m & \text{if } m = 1 \; \text{or} \; m \ge n \ge 2, \\
	4g + (2p-1)mn + 2mn - 2n & \text{if } n = 1 \; \text{or} \; n > m \ge 2.
\end{array}
\right.
\]

\end{Thm}

\begin{proof}
We only prove the case $m = 1$ or $m \ge n \ge 2$ since the proof of the case $n = 1$ or $n > m \ge 2$ is similar. First, we need to prove that our construction satisfies Condition \hyperref[cond:6]{6}. In our setup, we have
\[
f_{i,j} g_{k,w} = \tau^{-((i-1) + (p - 1 + j - k)mn + m(w-1))} + \sum_{\ell=v}^{\infty} \lambda_{\ell}^{(ij,kw)} \tau^{\ell},
\]
where $v = -\max\left\lbrace (i-1) + (j-1)mn, (p-k)mn + m(w-1) \right\rbrace +m(n-1)+1 \ge -(p-1)mn+1.$

We can obtain $\nu_P(f_{i,j} g_{j,w}) = d_{i,w} = -((i-1) + (p-1)mn + m(w-1))$. Additionally, $\nu_P(f_{i',j'} g_{k',w'}) = -((i'-1) + (p-1 + j' - k')mn + m(w'-1))$, which could be $d_{i,w}$ only if $(i', j' - k', w') = (i, 0, w)$, satisfying the condition stated in (\hyperref[eq:13]{13}).

Next, we consider the recovery threshold. Since $$\min \left\lbrace \nu_P(f_{i,j} g_{k,w}) \right\rbrace = -((m-1) + 2mn(p - 1) + m(n-1)) = -((2p-1)mn-1)$$ we have $f_{i,j} g_{k,w} \in \mathcal{L}((4g + 2(n-1)m)Q + ((2p-1)mn-1)P)$ for all $i, j, k, w$. Hence, $h = fg \in \mathcal{L}((4g + 2(n-1)m)Q + ((2p-1)mn-1)P)^{\frac{t}{m} \times \frac{s}{n}}$. Denote the $(i,w)$-th entry of $h$ as $h^{(i,w)}$, then for each $(i,w) \in \left[1, \frac{t}{m}\right] \times \left[1, \frac{s}{n}\right]$, $h^{(i,w)} \in \mathcal{L}((4g + 2(n-1)m)Q + ((2p-1)mn-1)P)$. Therefore, the recovery threshold $R = \deg((4g + 2(n-1)m)Q + ((2p-1)mn-1)P) + 1 = 4g + (2p-1)mn + 2mn - 2m$.
\end{proof}

\begin{Rmk}\label{rmk:5.5}
Using the same idea as in Construction \hyperref[con:4]{4}, we can extend Entangled Polynomial codes \cite{yu2020straggler} or Generalized PolyDot codes \cite{dutta2018unified} to the AG-based case. This extension will result in an AG-based DMM scheme in $\mathbb{F}_q$ with a recovery threshold of
\[
R = \left\{
\begin{array}{ll}
	4g + 3nmp - 2mp + p - 1 & \text{if } m = 1 \; \text{or} \; m \ge n \ge 2, \\
	4g + 3nmp - 2np + p - 1 & \text{if } n = 1 \; \text{or} \; n > m \ge 2.
\end{array}
\right.
\]
When $p=1$, the above recovery threshold is equal to that of Construction \hyperref[con:4]{4}, and when $p \ge 2$, the above recovery threshold is always larger than that of Construction \hyperref[con:4]{4}.

\end{Rmk}

In our setup, we can recover \( AB \) from the local expansion of \( h \) at the point \( P \).

\begin{Rmk} For \( m = 1 \) or \( m \ge n \ge 2 \), we have
\[
f_{i,j} g_{k,w} = \tau^{-((i-1) + (p - 1 + j - k)mn + m(w-1))} + \sum_{\ell=v}^{\infty} \lambda_{\ell}^{(ij,kw)} \tau^{\ell},
\]
where
\[
v = -\max\left\lbrace (i-1) + (j-1)mn, (p-k)mn + m(w-1) \right\rbrace + m(n-1) + 1 \ge -(p-1)mn + 1.
\]

Thus, the product of \( f_{i,j} \) and \( g_{k,w} \) will not affect the values of the coefficients of the term
\[
\left\lbrace \tau^{-((i-1) + (p-1)mn + m(w-1))} \right\rbrace_{(i,w) \in [1,m] \times [1,n]}
\]
in the local expansion of \( h \), since $-((i-1) + (p-1)mn + m(w-1)) < -(p-1)mn + 1$. Hence, we can recover \( AB \) from the coefficient of the term \( \tau^{-((i-1) + (p-1)mn + m(w-1))} \) in the local expansion of \( h \) at the point \( P \). This coefficient should correctly represent the values of \( C_{iw} = \sum_{j=1}^p A_{ij} B_{jw} \).

For \( n = 1 \) or \( n \ge m \ge 2 \), through a similar discussion as before, we can recover \( AB \) from the coefficient of the term \( \tau^{-((i-1)n + (p-1)mn + (w-1))} \) in the local expansion of \( h \) at the point \( P \). This coefficient should correctly represent the values of \( C_{iw} = \sum_{j=1}^p A_{ij} B_{jw} \).

The reason why we no longer use the non-special divisor \( D \) in our construction of AG-based PolyDot codes is that, when using the non-special divisor \( D \) as in Sections \hyperref[sec:3]{3} and \hyperref[sec:4]{4}, the product of \( f_{i,j} \) and \( g_{k,w} \) will influence the values of the coefficients of the term \( \tau^{-((i-1) + (p-1)mn + m(w-1))} \) (resp., \( \tau^{-(n(i-1) + (p-1)mn + (w-1))} \) ) in the local expansion of \( h \). Hence, we cannot correctly recover \( AB \) from the local expansion of \( h \).

\end{Rmk}

Next, we discuss the communication cost and per-worker computation complexity. Let \(N\) be the total number of worker nodes and \(R\) be the recovery threshold.

\begin{itemize}
    \item For communication cost, the master node sends \(O\left(\frac{tr}{mp} + \frac{rs}{np}\right)\) symbols to each worker node and receives \(O\left(\frac{ts}{mn}\right)\) symbols from each successful worker node. Therefore, the upload cost is \(O\left(\left(\frac{tr}{mp} + \frac{rs}{np}\right)N\right)\), and the download cost is \(O\left(\frac{ts}{mn} R\right)\).

    \item For worker computation complexity, each worker node computes matrix multiplications of size \(\frac{t}{m} \times \frac{r}{p}\) and \(\frac{r}{p} \times \frac{s}{n}\). Thus, the worker computation complexity is \(O\left(\frac{trs}{mpn}\right)\) using the native multiplication algorithm.
\end{itemize}

\section{Decoding Procedures} \label{sec:6}
In this section, we use the notations from previous sections unless specified otherwise. Under the general framework of AG-based DMM, we now show how to decode the product \(AB\) given evaluations of \(h = f_A f_B\) at \(R\) rational points.

Assume the matrix function \( h \) has entries in a Riemann-Roch space \(\mathcal{L}(G)\) for some divisor \( G \), and the set of evaluation points is \(\mathcal{P} = \{P_1, \ldots, P_R\}\). Then the corresponding AG-based code is \( C(G, \mathcal{P}) \). According to our constructions, \( AB \) is related to the coefficients of certain terms in the local expansion of \( h \) at the rational place \( P \). Thus, the decoding procedure consists of two steps:
\begin{itemize}
    \item[(i)] Compute the function \( h = f_A f_B \) from its evaluations at \( R \) rational places using the decoding algorithm for the algebraic geometry code \( C(G, \mathcal{P}) \).
    \item[(ii)] Recover \( AB \) from the local expansion of \( h \) at \( P \).
\end{itemize}

 Note that the recovery threshold is defined as the minimum number of evaluations required to uniquely determine \( h \). Thus, Step (i) will output the unique \( h \). For computational complexity analysis, we detail the decoding procedures for three types of constructions as follows.

%The trivial decoding algorithm is to solve linear equations as follows 

%Actually, $h$ is a matrix with entries in some Riemann-Roch space $\mathcal{L}(G)$ and each evaluation is a matrix with the same size of $h$. It suffice to solve the problem for the case that$h\in\mathcal{L}(G)$ is a $1\times 1$ matrix and $R$ evaluations are in $\F_q$. 
Assume \(\Lambda_1, \Lambda_2, \dots, \Lambda_K\) is a basis of \(\mathcal{L}(G)\) and their local expansions at \( P \) are
\[
\Lambda_i = \sum_{j=-v}^{\infty} \lambda^{(i)}_j \tau^j, \quad \text{for} \ i = 1, 2, \dots, K.
\]
Assume \( h = X_1 \Lambda_1 + X_2 \Lambda_2 + \dots + X_K \Lambda_K \),
where each coefficient \( X_i \) is a matrix over \(\mathbb{F}_q\) with the same size as \( h \). Without loss of generality, assume we are given \( R \) evaluations \( h(P_1), \dots, h(P_R) \). Then
\[
(X_1, X_2, \dots, X_K)
\begin{pmatrix}
    \Lambda_1(P_1) & \Lambda_1(P_2) & \cdots & \Lambda_1(P_R) \\
    \Lambda_2(P_1) & \Lambda_2(P_2) & \cdots & \Lambda_2(P_R) \\
    \vdots & \vdots & \ddots & \vdots \\
    \Lambda_K(P_1) & \Lambda_K(P_2) & \cdots & \Lambda_K(P_R)
\end{pmatrix} = (h(P_1), h(P_2), \dots, h(P_R)),
\tag{14}\label{eq:14}
\]
where the product $X_i\cdot \Lambda_i(P_j)$ is a scalar multiplication of matrix. By solving the above linear equations of matrices, we can get $h=\sum_{i=1}^K X_i\Lambda_i$. Thus the local expansion of $h$ can be obtained from the local expansions of $\Lambda_1,\Lambda_2,\dots,\Lambda_K$.  
\begin{enumerate}
    \item[(1)]
   In our AG-based Polynomial codes, the above divisor is \( G = 2D + (mn - 1)P \) in Construction \hyperref[con:1]{1} and \( G = D + (mn - 1)P \) in Construction \hyperref[con:2]{2}, and \( h \in \mathcal{L}(G)^{\frac{t}{m} \times \frac{s}{n}} \). Since \(\{ f_i g_j \in \mathcal{L}(G) \mid i = 1, \dots, m, j = 1, \ldots, n \}\) are linearly independent over \(\mathbb{F}_q\), we can expand this set to form the basis \(\Lambda_1, \dots, \Lambda_K\) of \(\mathcal{L}(G)\) such that \( \Lambda_{(i-1) + (j-1)m} = f_i g_j \) for \( i = 1, \ldots, m \) and \( j = 1, \ldots, n \). Therefore,
   \[
 A_i B_j =X_{(i-1) + (j-1)m},\ \text{for} \ 1\leq i \leq m,\ 1\leq j \leq n. \tag{15}\label{eq:15}
   \]

\item[(2)] In our AG-based MatDot codes, the above divisor is \( G = 2D + 2(p-1)P \) and \( h \in \mathcal{L}(G)^{t \times s} \). 
The product \( AB = \sum_{i=1}^p A_i B_i \) is the coefficient of \( \tau^{-(p-1)} \) in the local expansion of \( h \). Therefore, 
\[
 AB = \sum_{i=1}^K X_i \lambda^{(i)}_{-(p-1)}, \tag{16}\label{eq:16}
\]

where \( \lambda^{(i)}_{-(p-1)} \) is the coefficient of \( \tau^{-(p-1)} \) in the local expansion of \( \Lambda_i \).

\item[(3)] In the AG-based PolyDot codes, the above divisor is \( G = (4g + 2c)Q + ((2p - 1)mn - 1)P \) and \( h \in \mathcal{L}(G)^{\frac{t}{m} \times \frac{s}{n}} \), where
\[
c = \left\{
\begin{array}{ll}
m(n-1) & \text{if } m = 1 \; \text{or} \; m \ge n \ge 2, \\
n(m-1) & \text{if } n = 1 \; \text{or} \; n > m \ge 2.
\end{array}
\right.
\]
The product \( AB = (C_{i,w})_{1\leq i\leq m, 1\leq w\leq n} \), where \( C_{i,w} = \sum_{j=1}^{p} A_{i,j} B_{j,w} \).

\begin{itemize}
    \item If \( m = 1 \) or \( m \ge n \ge 2 \), \( C_{i,w} \) is the coefficient of \( \tau^{-((i-1) + (p-1)mn + m(w-1))} \) in the local expansion of \( h \). Therefore, 
   \[
    C_{i,w} = \sum_{i=1}^K X_i \lambda^{(i)}_{-((i-1) + (p-1)mn + m(w-1))},\tag{17}\label{eq:17}
   \]

    \item If \( n = 1 \) or \( n > m \ge 2 \), \( C_{i,w} \) is the coefficient of \( \tau^{-((i-1)n + (p-1)mn + (w-1))} \) in the local expansion of \( h \). Therefore, 
\[
C_{i,w} = \sum_{i=1}^K X_i \lambda^{(i)}_{-((i-1)n + (p-1)mn + (w-1))},
\tag{18}\label{eq:18}
\]
\end{itemize}
where the above $\lambda^{(i)}_{-*}$ is the coefficient of $\tau^{-*}$ in the local expansion of $\Lambda_i$ at $P$.
\end{enumerate}

\begin{Rmk}[Decoding complexity]
In Step (i), we decode \( h \) from its \( R \) evaluations by solving the linear equations (\hyperref[eq:14]{14}). The coefficient matrix \( G = (\Lambda_i(P_j))_{1 \leq i \leq K, 1 \leq j \leq R} \) is a full row rank matrix, hence it has a right inverse \( G^{-1} \). Thus, we have
\[
(X_1, X_2, \ldots, X_K) = (h(P_1), h(P_2), \ldots, h(P_R)) G^{-1},
\]
where the product of $h(P_i)$ and the $(i,j)$-th entry of $G^{-1}$ is the scalar multiplication. The complexity of computing \( G^{-1} \) is \( O(K^2 R) \) using Gaussian elimination. 
\begin{enumerate}
    \item[(1)] For AG-based Polynomial codes, we only need the first \( mn \) values \( X_1, \ldots, X_{mn} \), which will cost $O(mn\frac{t}{m}\frac{s}{n}R)=O(tsR)$ operations in $\F_q$. By (\hyperref[eq:15]{15}), the total complexity to recover $AB$ is \( O(K^2 R+ts R)\).
    \item[(2)] For AG-based MatDot codes, we need all the values of \( X_1, \ldots, X_K \), which will cost \( O(K^2 R + ts K R) \) operations. By (\hyperref[eq:16]{16}), the total complexity to recover $AB$ is \( O(K^2 R+ts KR+tsK)=O(tsKR + K^2 R)\).
    
    \item[(3)] For AG-based PolyDot codes, we also need all the values of \( X_1, \ldots, X_K \). By (\hyperref[eq:17]{17}) and (\hyperref[eq:18]{18}), the total complexity to recover $AB$ is \( O\left(\frac{ts}{mn} K R + K^2 R\right) \).
\end{enumerate}

\end{Rmk}

\section{Comparisons} \label{sec:7}
In this section, we compare the recovery thresholds of our AG-based Polynomial and MatDot codes with those in \cite{fidalgo2024distributed}.
For the function field \(F/\mathbb{F}_q\), recall that \(D \in \text{Div}(F)\) is a positive non-special divisor with \(\deg(D) = g\) and \(\ell(D) = 1\). Let \(P \in \mathbb{P}_F\) be a rational point such that \(P \not\in \text{supp}(D)\). Denote \(W(P)\) as the Weierstrass semigroup of \(P\). Define its conductor \(c(P)\) as
\[
c(P) := \min \left\lbrace k \in W(P) : [k, \infty) \subseteq W(P) \right\rbrace.
\]
By the Weierstrass gap theorem introduced in Section \hyperref[sec:2.1]{2.1}, we have \(g + 1 \le c(P) \le 2g\) for \(g > 0\).

Observe that our AG-based Polynomial and MatDot codes consider functions in the Riemann-Roch space \(\mathcal{L}(D + \infty P) = \bigcup_{i=0}^{\infty} \mathcal{L}(D + iP)\), whereas the constructions in \cite{fidalgo2024distributed} consider functions in the one-point Riemann-Roch space \(\mathcal{L}(\infty P) = \bigcup_{i=0}^{\infty} \mathcal{L}(iP)\).

\subsection{Comparison with AG-Based Polynomial Codes}

For AG-based polynomial codes, \cite{fidalgo2024distributed} provides three constructions with the following recovery thresholds:
\begin{table}[H]
	\centering
	\begin{tabular}{|c|c|c|}
		\hline
		& $m \not\in W(P)$ & $m \in W(P)$ \\
		\hline
		Construction $A$ in \cite{fidalgo2024distributed} & $2c(P) + mn$ & $2c(P) + mn$ \\
		\hline
		Construction $B$ in \cite{fidalgo2024distributed} & $c(P) + m'n$ & $c(P) + mn$ \\
		\hline
		Construction $C$ in \cite{fidalgo2024distributed} & $c(P) + m_n+m$ & $c(P) + mn$ \\
		\hline
	\end{tabular}
	\caption{Recovery thresholds of AG-based polynomial codes in \cite{fidalgo2024distributed}}
	\label{tab:4}
\end{table}

\noindent where \( m' := \min \left\lbrace k \in W(P) : k \ge m \right\rbrace \), \( m_1 := 0 \), and
\( m_i := \min \{ k \in W(P) : k \ge m_{i-1} + m \} \) for \( i = 2, \ldots, n \).

When \( m \in W(P) \), our Construction \hyperref[con:2]{2} has a recovery threshold \( R = g + mn \), which is always better than \( c(P) + mn \) since \( g + 1 \le c(P) \le 2g \). Hence, we consider the case when \( m \notin W(P) \). Comparing our Construction \hyperref[con:1]{1} with Construction A in \cite{fidalgo2024distributed}, our recovery threshold \( R = 2g + mn \) is better than \( 2c(P) + mn \) since \( g + 1 \le c(P) \le 2g \).

To compare our Construction \hyperref[con:2]{2} with Constructions B and C in \cite{fidalgo2024distributed}, we need to make some modifications:
Since \( m' \in W(P) \), then \( 2m', \ldots, (n-1)m' \) are also in \( W(P) \). Thus, we can modify Construction \hyperref[con:2]{2} to obtain the following construction:

\begin{Con} \label{con:5} [The First Variation of AG-based Polynomial Code]

\noindent Set $f_i$ the same as in Construction \hyperref[con:1]{1} for $i = 1, 2, \ldots, m$ and set $g_1 = 1 \in \mathbb{F}_q$ and $g_j \in \mathcal{L}((j-1)m'P) \backslash \mathcal{L}(((j-1)m'-1)P)$ for $j = 2, \ldots, n$, i.e.,
	\[
	g_j = \tau^{-(j-1)m'} + \sum_{\ell=-(j-1)m' + 1}^{\infty} \hat{\lambda}_\ell^{(j)} \tau^\ell \quad \text{for} \quad j = 2, \ldots, n,
	\]
	
\end{Con}

Similarly, since $m_1, \ldots, m_n \in W(P)$, we can obtain the following construction:
\begin{Con} \label{con:6} [The Second Variation of AG-based Polynomial Code]
	
 \noindent Set $f_i$ the same as in Construction \hyperref[con:1]{1} for $i = 1, 2, \ldots, m$ and set $g_1 = 1 \in \mathbb{F}_q$ and $g_j \in \mathcal{L}(m_jP) \backslash \mathcal{L}((m_j-1)P)$ for $j = 2, \ldots, n$, i.e.,
	\[
	g_j = \tau^{-m_j} + \sum_{\ell=-m_j + 1}^{\infty} \hat{\lambda}_\ell^{(j)} \tau^\ell \quad \text{for} \quad j = 2, \ldots, n.
	\]

\end{Con}

For Constructions \hyperref[con:5]{5} and \hyperref[con:6]{6}, we have the following theorem, the proof of which is similar to the proof of Theorem \hyperref[thm:3.2]{3.2}:

\begin{Thm}\label{thm:7.1}
For any prime power \( q \) and partitioning parameters \( m, n \ge 1 \) that divide \( t \) and \( s \) respectively, Construction \hyperref[con:5]{5} provides an AG-based DMM scheme in \( \mathbb{F}_q \) with a recovery threshold of \( R = g + m'n \), and Construction \hyperref[con:6]{6} provides an AG-based DMM scheme in \( \mathbb{F}_q \) with a recovery threshold of \( R = g + m_n + m \), where \( m' = \min \left\lbrace k \in W(P) : k \ge m \right\rbrace \) and \( m_1 = 0 \), \( m_i = \min \left\lbrace k \in W(P) : k \ge m_{i-1} + m \right\rbrace \) for \( i = 2, \ldots, n \).

\end{Thm}

Hence, our Constructions \hyperref[con:5]{5} and \hyperref[con:6]{6} have better recovery thresholds than Constructions $B$ and $C$ in \cite{fidalgo2024distributed}.

\subsection{Comparison with AG-Based MatDot Codes}

For AG Matdot codes, \cite{fidalgo2024distributed} provides a construction with a recovery threshold of \( 2c(P) + 2p - 1 \). Therefore, it is clear that our Construction \hyperref[con:3]{3} offers a better recovery threshold of \( R = 2g + 2p - 1 \). Moreover, \cite{fidalgo2024distributed} also provides an optimal recovery threshold for \( p \ge 2c(P) \) when using functions in the one-point Riemann-Roch space, which is hard to determine because it depends on the structure of \( W(P) \). However, we can compare the recovery threshold of our Construction \hyperref[con:3]{3} with the optimal recovery threshold for some specific curves.

\begin{comment}
	\begin{Thm}\label{thm:5.2}
		Let $\delta \in [0, c(P)] \cap W(P)$, and define $n(\delta) := \left|[ \delta, c(P) - 1 ] \cap W(P) \right|$. Let $m \ge 2c(S)$. Consider an element $\delta \in [0, c(P)] \cap W(P)$ that maximizes $\delta + 2n(\delta)$. Then the optimal recovery threshold is
		
		\[
		2p - 1 + 4c(P) - 2(\delta + 2n(\delta)).
		\]

	\end{Thm}
	
	It seems difficult to determine the optimal recovery threshold since $\delta + 2n(\delta)$ depends on the structure of $W(P)$. However, we can compare the recovery threshold of our Construction \hyperref[con:5]{5} with the optimal recovery threshold for some specific curves.
	\begin{Exm}
		For the Hermitian curve, $c(P) = 2g$ and $\max\left\lbrace \delta + 2n(\delta) \right\rbrace \approx \frac{5}{2}g$ by \cite[Proposition 4]{fidalgo2024distributed}. Therefore, for Hermitian curves, the above recovery threshold is approximately $2p - 1 + 3g$, while our recovery threshold is $2p - 1 + 2g$, which is better.

	\end{Exm}
	
	\begin{Exm}
		For the elliptic curve, $c(P) = g + 1$ and $\max\left\lbrace \delta + 2n(\delta) \right\rbrace = g + 1$ by \cite[Proposition 3]{fidalgo2024distributed}. Therefore, for elliptic curves, the above recovery threshold is $2p - 1 + 2g + 2$, while our recovery threshold is $2p - 1 + 2g$, which is better.

	\end{Exm}
\end{comment}

\begin{Exm}
For Hermitian curves, the recovery threshold in \cite{fidalgo2024distributed} is approximately $2p - 1 + 3g$ according to \cite[Proposition 4]{fidalgo2024distributed}, while our recovery threshold is $2p - 1 + 2g$, which is better.

\end{Exm}

\begin{Exm}
For elliptic curves, the above recovery threshold in \cite{fidalgo2024distributed} is $2p - 1 + 2g + 2$ according to \cite[Proposition 3]{fidalgo2024distributed}, while our recovery threshold is $2p - 1 + 2g$, which is better.

\end{Exm}

\section{Conclusion}
In this paper, we apply algebraic geometry codes to distributed matrix multiplication. By utilizing local expansions in function fields, we extend Polynomial codes \cite{yu2017polynomial}, Matdot codes \cite{dutta2019optimal}, and PolyDot codes \cite{dutta2019optimal}, as well as Entangled Polynomial codes \cite{yu2020straggler} and Generalized PolyDot codes \cite{dutta2018unified}, to the AG-based case. This extension overcomes the limitations on the size of the finite field \( q \) imposed by previous RS-based DMM schemes. Moreover, compared to the previous AG-based DMM schemes in \cite{fidalgo2024distributed}, our schemes demonstrate superior recovery thresholds.

\section*{Acknowledgment}
The work of Jiang Li was partially supported by the National Key Research and Development Program under Grant 2022YFA1004900.
The work of Songsong Li was supported in part by the National Natural Science Foundation of China under Grant 12101404. The work of Chaoping Xing was supported by the National Natural Science Foundation of China under Grant 12031011 and 123611418.

\bibliographystyle{plain}
\bibliography{ref}

\end{document}